\documentclass{article}

\usepackage[final]{neurips_2024}
\usepackage{preamble}
\usepackage{notation}
\usepackage[utf8]{inputenc} % allow utf-8 input
\usepackage[T1]{fontenc}    % use 8-bit T1 fonts
\usepackage{hyperref}       % hyperlinks
\usepackage{url}            % simple URL typesetting
\usepackage{booktabs}       % professional-quality tables
\usepackage{amsfonts}       % blackboard math symbols
\usepackage{nicefrac}       % compact symbols for 1/2, etc.
\usepackage{microtype}      % microtypography

\title{Unravelling in Collaborative Learning}

\author{%
  Aymeric Capitaine\textsuperscript{1} \\
  \And
  Etienne Boursier\textsuperscript{2} \\
  \AND
  Antoine Scheid\textsuperscript{1} \\
  \And
  Eric Moulines\textsuperscript{1} \\
  \And
  Michael I. Jordan\textsuperscript{3} \\
  \And
  El-Mahdi El-Mhamdi\textsuperscript{1} \\
  \AND
  Alain Durmus\textsuperscript{1}
}

\begin{document}

\maketitle
{\centering
\noindent\textsuperscript{1} Centre de Mathématiques Appliquées – CNRS – École polytechnique – Palaiseau, 91120, France \\
\noindent\textsuperscript{2} INRIA Saclay, Universit\'e Paris Saclay, LMO - Orsay, 91400, France\\
\noindent\textsuperscript{3} Inria, Ecole Normale Sup\'erieure, PSL Research University - Paris, 75, France \newline
\newline
}

\everypar{\looseness=-1}
\begin{abstract} Collaborative learning offers a promising avenue for leveraging decentralized data. However, collaboration in groups of strategic learners is not a given. In this work, we consider strategic agents who wish to train a model together but have sampling distributions of different quality. The collaboration is organized by a benevolent aggregator who gathers samples so as to maximize total welfare, but is unaware of data quality. This setting allows us to shed light on the deleterious effect of \textit{adverse selection} in collaborative learning. More precisely, we demonstrate that when data quality indices are private, the coalition may undergo a phenomenon known as \emph{unravelling}, wherein it shrinks up to the point that it becomes empty or solely comprised of the worst agent. We show how this issue can be addressed without making use of external transfers, by proposing a novel method inspired by probabilistic verification. This approach makes the grand coalition a Nash equilibrium with high probability  despite information asymmetry, thereby breaking unravelling.
\end{abstract}
\section{Introduction}
Collaborative learning is a framework in which multiple agents share their data and computational resources to address a common learning task \citep{blum-etal,kairouz2021advances}. A significant challenge arises when the quality of these distributions is unknown centrally and agents are strategic. Indeed, participants may be tempted to withhold or misrepresent the quality of their data to gain a competitive advantage. These  strategic behaviors and their consequences have been studied extensively in the literature on information economics \citep{mc95, lm01}. In particular, information asymmetry is known to result in \textit{adverse selection}, whereby low-quality goods end up dominating the market. In the current paper we study collaborative learning from the perspective of information economics.

\looseness=-1
A vivid illustration of adverse selection is found in Akerlof's seminal work on the market for lemons \citep[second-hand cars of low quality,][]{a70}. Because buyers cannot properly assess the quality of cars on the second-hand market, their inclination to pay decreases. As a consequence, sellers with high-quality cars withdraw from the market, since the proposed price falls below their reservation price. This in turn lowers the buyers' expectation regarding the average quality of cars on the market, so their willingness to pay decreases even more, which \textit{de facto} crowds out additional cars. The market may therefore enter a death spiral, up to the point where only low-quality cars are exchanged in any competitive Nash equilibrium. This phenomenon is known as \textit{unravelling}. The insurance market serves as another poignant example of this effect \citep{rs76,ef11,h13}. In insurance, information asymmetry arises due to the fact that insurees possess private knowledge about their individual risk profiles, which insurers lack. Individuals with higher risk are more inclined to purchase policy, while those with lower risk opt out. Consequently, insurers are left with a pool of policyholders skewed towards higher risk, leading to increased premiums to cover potential losses. This, in turn, prompts low-risk individuals to exit the market, exacerbating adverse selection further---a cycle reminiscent of the unravelling described by Akerlof.

We study the problem of whether collaborative learning could also fall victim to unravelling. We consider strategic agents who have access to sampling distributions of varying quality and wish to jointly train a  model.  They delegate the training process to a central authority who collects samples so as to maximize total welfare. We ask whether adverse selection can arise when data quality is private information. In particular, the presence of a low-quality data owner may harm the model, prompting high-quality owner to leave the collaboration and train a model entirely on their own. Their departure would decrease the average data quality even more, and create a vicious circle. In the worst case, the coalition of learners would reduce to the lowest data quality owner alone. This question is of prime importance from a practical point of view, because unravelling could jeopardize the long-run stability of collaborative models deployed at large scale. 

Our contribution is threefold:
\begin{enumerate}[leftmargin=1cm]
    \item We provide a rigorous framework for analyzing  collaborative learning with strategic agents having data distributions of varying quality. On the one hand, we leverage tools from domain adaptation to capture a notion of data quality formally. On the other hand, we model collaboration as a principal-agent problem, where the principal is an aggregator in charge of collecting samples so as to maximize social welfare. This setup allows us to derive the benchmark welfare-maximizing collaboration scheme when data quality is public information.
    \item We show that when data quality is  private, a naive aggregation strategy which consists in asking agents to declare their quality type and applying the optimal scheme results in a complete unravelling. More precisely, the set of agents willing to collaborate is either empty or made of the lowest-quality data owner alone at any pure Nash equilibrium. 
    
    \item \looseness=-1 We present solutions to unravelling. When transfers are allowed, the VCG mechanism suffices to re-establish optimality. When transfers are not possible, we leverage probabilistic verification techniques  
    to design a mechanism which breaks unravelling. More precisely, we ensure that the optimal, grand coalition ranks with high probability among the Nash equilibria of the game induced by our mechanism. We demonstrate how to implement our mechanism practically in the setting of classification. 
\end{enumerate}  

\paragraph{Related work.}\looseness=-1 The issue of information asymmetry in machine learning has been an area of recent activity. Several learning settings have been considered, including bandits \citep{wl24}, linear regression \citep{kk21,kk212}, classification \citep{bhp21} and empirical risk minimization \citep{dk23,lst23} in a federated context \citep{xkn22}. 

Most of these studies focus on the sub-problem of \textit{moral hazard}, where agents take actions that are unobserved by others. This situation usually results in under-provision of effort and inefficiency at the collective scale \citep{lm01}. This issue appears naturally in federated learning, because model updates are performed locally. For instance,  \citet{kwj22} show that heterogeneity in sampling costs results in total free-riding without a proper incentive scheme. \citet{hk22} show that under-provision of data points in federated learning arises from privacy concerns. \citet{yx23} consider a federated classification game where agents can reduce the noise in their data distribution but incur a costly effort to do so. In the same vein, \citet{hkm23} study the case of agents who are interested in different models, and may skew their sampling measure accordingly. \citet{saig2023delegated} and \citet{abjh23} study hidden actions when a principal delegates a predictive task to another agent, and show that thresholds or linear contracts are able to approximate the optimal contracts. 

\textit{Adverse selection} is another type of information asymmetry, where preferences of agents are unobserved rather than actions. This issue also naturally arises in collaborative learning, because the data distributions from which agents sample or about which they care may not be public. While data heterogeneity is a widely explored topic in federated learning (see for instance \citealt{gao2022survey} for a general survey, and \citealt{fu2023client} for the specific problem of \textit{client selection}), the strategic aspect has been rarely considered. Most studies doing so focus on hidden sampling costs (see, e.g., \citealt{kwj22}, or \citealt{wl24} in a bandit context), but few address the fundamental problem of distribution shift. \citet{abjh23} mentions the issue of adverse selection when a principal delegates a predictive task to an agent, and provide qualitative insights about the optimal contract. However, they consider a single agent and leave aside the question of participation. Finally, \citet{w24} and \citet{toik23} study the question of the stability of collaborative learning between competing agents. However, their analysis relies on ad-hoc market structures rather than information asymmetry per se. As such, our work is the first to demonstrate the effect of imperfect information on the sustainability of collaborative learning. 

\paragraph{Organization.} \Cref{section:model} presents our model and assumptions. \secref{section:full_information} studies a full information benchmark, which allows us to derive the welfare-maximizing contribution scheme. In \secref{section:hidden_information}, we turn to the more realistic case where data quality is private information. We first show that in this case, a naive aggregation method leads to total unravelling. Second, we introduce a mechanism which breaks unravelling by  inducing a game where the grand coalition is a pure strategy Nash equilibrium with high probability.

\section{Model}\label{section:model}

\paragraph{Statistical framework.}Let $(\msx,\mcx)$ and $(\msy,\mcy)$ be two measurable spaces and denote by $\mcp$ a family of  probability measures on $(\msx\times \msy,\mcx\otimes \mcy)$. 
We consider $\agents=\iint{1}{J}$ agents who aim to perform a prediction task associated with a hypothesis class $\hypoclass \subset  \{ g : \msx \to \msy\}$, a loss function $\ell : \msy \times \msy$ and a probability measure $P_0 \in\mcp$. Each agent seeks to minimize $g\in\hypoclass \mapsto \risk{g}{P_0}$ where for any probability distribution $P \in \mcp$, $\risk{g}{P}= \int  \ell(g(x), y) \rmd P(x,y)$ is the risk associated to $g \in\hypoclass$ with respect to $P$.

We leverage tools and results from statistical learning theory. Denote by $\{(X_i,Y_i)\}_{i \geq 1}$ the canonical process on $\msx \times \msy$ and 
denote by $\P_P$ and $\PE_P$ the canonical probability and expectation under which $\{(X_i,Y_i)\}_{i \geq 1}$ are \iid~random variables with distribution $P \in \mcp$.
With this notation, we can introduce our assumptions on $\hypoclass$ and~$\cP$.

\begin{assumption}\label{assumption_upper_bound_risk}
    For $\delta\in(0,1)$, there exist $\alpha_{\delta} >0$, $\beta>0$ and $\gamma>0$ such that for any distribution $P\in\cP$,  hypothesis $g\in\cG$ and $n\geq 1$
    $$
   \PPArg{\abs{\risk{g}{P}-\frac{1}{n}\sum_{i=1}^{n}\ell(g(X_i), Y_i)}\leq \frac{\alpha_{\delta}}{(1+n)^\gamma} + \beta}{P} \geq 1 - \delta \eqsp.
    $$
  \end{assumption}
  
This assumption covers a wide range of situations. For instance in the classification case where $\msy=\defEns{-1,1}$ and $\ell: (y,\tilde{y})\mapsto \2{y\tilde{y}\leq 0}$, $\alpha_{\delta}=\sqrt{\ln(1/\delta)/2}$, $\beta=2\textsc{Rad}(\hypoclass)$ and $\gamma=1/2$ where $\textsc{Rad}(\hypoclass)$ is the   empirical Rademacher complexity of $\hypoclass$ \citep{b04}. Similarly, the Bayesian PAC approach in the linear regression context with bounded loss leads to $\alpha_{\delta}=\KL(\rho \|\pi)+\ln(1/\delta),\,\beta=\norm{\ell}^2 _{\infty} / 8$ and $\gamma=1$ where $\rho$ and $\pi$ are any distributions on  $\hypoclass$ \citep[][Corollary 4]{sfp19}. 

For ease of notation, we let $\risk{g}{j}$ serve as a shorthand for $\risk{g}{P_j}$.
It is moreover assumed that agent $j\in\agents$ cannot directly sample from $\Pzero$, but has instead access to a distribution $P_j\in\cP$ which deviates from $\Pzero$ according to the $\hypoclass$-divergence:
\begin{assumption}
    \label{definition_type}
For any $ j \in [J]$, $P_j \in \cP$ has finite  $\hypoclass$-divergence:  $ \type_j = \sup_{g\in\hypoclass} \absLigne{\risk{g}{j} - \risk{g}{0}} < \plusinfty$ \eqsp.
\end{assumption} 
\looseness=-1 Intuitively, for any $j\in\agents,\,\type_j\geq 0$ models the bias incurred by having access to samples from $P_j$ instead of the target distribution $P_0$. More precisely, the risk excess associated with empirical risk minimization (ERM) based on samples from $P_j$ is in the worst case at least $\type_j$.
A poor sampling distribution $P_j$---which corresponds to a high $\type_j$ in the previous expression---might be the consequence of low-quality sensors or degraded experimental conditions resulting in noisier data points. 

The class of discrepancies appearing in  \Cref{definition_type} has been considered in the domain adaptation literature~\citep[see, e.g.,][]{bsb10,kb04,k19}. It provides a natural framework to analyze the behavior of models trained on diverse distributions, and is practically appealing since it can be easily estimated in the context of classification \citep{bsb10}.

\looseness=-1 We make the  following assumptions on the  quality indexes $(\type_1,\ldots,\type_J)$, hereafter referred to as \textit{types}.%
\begin{assumption}
    \label{assumption_types_sorted}There exists $(\typemin, \typemax)\in\R_{+}^2$ such that $\typemin\leq \type_1 < \type_2 < \ldots < \type_J \leq \typemax$.
\end{assumption}
\Cref{assumption_types_sorted} The ordering assumption is just for ease of exposition but is neither used by the aggregator nor the agents. Our condition that types are strictly  different is for convenience in simplifying the proofs.
\paragraph{Collaborative learning framework.}We further suppose that agent $j\in\agents$ can fit a model $g\in\hypoclass$ based on \iid~samples $\defEnsLigne{(X^j _1,Y^j _1),\ldots,(X^j _{n_{j}},Y^j _{n_{j}})}$
  of size $n_{j} \geq 0$ from $P_j \in\cP$ in one of two ways: they can compute on their own, or collaborate. This is captured by the two following options:
\begin{enumerate}[wide, labelwidth=!, labelindent=0pt, label= \textbf{Option} \arabic*)]
\item \label{item:option1}  Agent $j$ performs ERM on their own samples:

\begin{equation}
    \label{definition_outside_model}
 \txts   \agentmodel{j} = \argmin_{g\in\hypoclass}\emprisk{g}{j}\eqsp , \qquad  \emprisk{g}{j}=n^{-1}_j \sum_{i=1}^{n_{j}}\ell(g(X^j _i),Y^j _i) \eqsp.
  \end{equation}
  
  This non-collaborative procedure is referred to as the \emph{outside option}.
\item \label{item:option2} Agent $j$  can take part in a \textit{coalition} orchestrated by a central data aggregator, encoded as $\collab=(B_1,\ldots, B_J)\in\defEns{0,1}^J$, where $B_j=1$ means that agent $j$ is member of the coalition. We also write $\cB=\defEns{j\in\agents:\:B_j=1}$. In exchange for their samples, agent $j$ gains access to the collaborative model trained over the concatenation of samples:

\begin{align}
    \collabmodel&=\argmin_{g\in\hypoclass}\emprisk{g}{\collab}\label{definition_collab_model}\eqsp, \\
    \text{ where }\quad \emprisk{g}{\collab}&= N^{-1}\sum_{j\in\agents}B_j \sum_{i=1}^{n_{j}}\ell(g(X^{j}_i), Y^{j}_i)\quad\text{ and } \quad N = \sum_{j\in\agents}B_jn_{j} \notag\eqsp.
\end{align}

\end{enumerate}
\paragraph{Agent utilities.}
We assume that agents incur a unitary cost for sampling from their distribution and dislike statistical risk. Therefore, a baseline model for measuring the preferences associated with a model $g \in \mcg$ and a number of samples $n$ is based on a linear map:  $(g,n) \mapsto -a\risk{g}{0}-cn$ for~$a,c>0$. In practice, however, $\risk{g}{0}$ is typically unknown so agents instead can use a PAC bound of the form $\P(\risk{\hat{g}_n}{0}\leq \Rstar + \eps)\geq 1 -\delta$ to a assess a model $\hat{g}_n\in\hypoclass$ trained over their samples,
where $\eps>0$, $\delta\in(0,1)$ and $\Rstar=\inf_{g\in\hypoclass}\risk{g}{0}$. Our next result shows that  \cref{assumption_upper_bound_risk} and \cref{definition_type} allow each agent to pin down such an $\eps > 0$, under either \ref{item:option1} or \ref{item:option2}.
Assuming \Cref{assumption_upper_bound_risk}, 
 define the function $\riskexcess$ for any $(\theta,n) \in \types\times\R_{+}$ by%
\begin{equation}
  \label{eq:def_vareps_delta}
\riskexcess(\type,n) =  2\parentheseDeux{\alpha_{\delta}(1+n)^{-\gamma} + \beta + \type} \eqsp.
\end{equation}%
\begin{restatable}{lemma}{excessrisk}\label{lemma:excessrisk}
  Assume \Cref{assumption_upper_bound_risk} and \Cref{definition_type}.
    \begin{enumerate}[wide, labelwidth=!, labelindent=0pt,label=(\roman*)]
    \item 
  Any agent $j\in\agents$ picking the outside \ref{item:option1} obtains a model $\agentmodel{j}\in\hypoclass$ achieving
    $$
    \risk{\agentmodel{j}}{0} \leq \Rstar + \riskexcess(\type_j ,n_{j})\quad\text{with probability }1-\delta\eqsp.
    $$
    \item   Any coalition $\collab\in\defEns{0,1}^J$ drawing $\mathbf{n} = (n_1,\ldots,n_{J}) \in \R_{+}^J$, samples obtains a model $\collabmodel\in\hypoclass$ achieving
    $$
    \risk{\collabmodel}{0} \leq \Rstar + \riskexcess(\vartheta, N)\quad\text{with probability }1-\delta\eqsp,
    $$
    where $N$ is the total of samples and $\vartheta$ is the weighted average type within the coalition:
    \begin{equation}
      \label{eq:def_N_theta}
      N(B,\nbf)=\sum_{j\in\agents}B_jn_{j} \eqsp, \qquad \avgtype = N^{-1}\sum_{j \in \agents} B_jn_{j} \type_j \eqsp.
    \end{equation}
  \end{enumerate}
\end{restatable}

When the context is clear, we write $\riskexcess(\collab, \bfn)=\riskexcess(\vartheta(\collab, \bfn), N(\collab, \bfn))$ to lighten notation. Based on \Cref{lemma:excessrisk}, we define the utility of agent $j\in\agents$ as
\begin{equation}
    \label{definition_utility}
    \utility_{j}: (\collab, \bfn)\mapsto -a\parentheseDeux{\,\Rstar +(1-B_j)\riskexcess(\type_j,n_{j}) + B_j \riskexcess(\collab,\bfn)\,}-cn_{j}\eqsp.
  \end{equation}

Note that for any $j\in \agents$, we take $n_{j}\geq 0$ to be a real number for ease of presentation.  In \eqref{definition_utility}, $a/c >0$ captures the extent to which individuals are willing to trade off model quality against sampling cost.
  
  For any $j\in\agents$ and $\collab_{-j}=(B_1,\ldots,B_{j-1},B_{j+1},\ldots,B_J)\in\defEns{0,1}^{J-1}$, we denote by a slight abuse of notation $(B, \collab_{-j})=(B_1,\ldots,B_{j-1},B,B_{j+1},\ldots,B_J)$ for any $B\in\defEns{0,1}$. Similarly for $\bfn_{-j}=(n_1,\ldots, n_{j-1},n_{j+1},\ldots,n_{j})\in\R_{+}^{J-1}$, we write $(n,\bfn_{-j})=(n_1,\ldots,n_{j-1},n,n_{j+1},\ldots,n_{j})$ for any $n\geq 0$.
  We can then characterize the optimal behavior of any agent picking the outside option as follows.

\begin{restatable}{proposition}{outsideoption}
  \label{prop:outsideoption}
  Assume \Cref{assumption_upper_bound_risk} and  \Cref{definition_type}.
  For any $j\in\agents$, $\collab_{-j}\in\defEns{0,1}^{J-1}$ and $\bfn_{-j}\in\R_{+}^{J-1}$, the optimal number of samples to draw under \ref{item:option1} is $\argmax_{n\ge 0}\utility_{j} ((0,\collab_{-j}),(n,\bfn_{-j})\,;\,\type_j)= \noutside$ where
  \begin{equation}
    \label{eq:def_noutside}
    \noutside = (2ac^{-1}\gamma\alpha_{\delta})^{1/(\gamma+1)} -1 \eqsp.
  \end{equation}
\end{restatable}
In what follows, we assume $2a/c > (\gamma\alpha_{\delta})^{-1}$ to exclude the pathological case where no agent is willing to sample data points. From now on, we denote by
\begin{equation}
  \label{eq:def_outside_j}
  \outside _j = \utility_{j} ((0,\collab_{-j}),(\noutside,\bfn_{-j})) \eqsp
\end{equation}
the best achievable utility under the outside option. Note that  $\noutside$ does not depend on $\type_j\in\types$ (but $\outside_j$ does) so all agents outside of the coalition draw a same number of data points $\noutside >0$. This result, which may be surprising at first glance, comes from the fact that all agents have the same accuracy-to-sampling-cost ratio $a/c$ in their utility.
\paragraph{Aggregator.}We finally assume that the aggregator acts benevolently to set up a Pareto-optimal collaboration, by maximizing the total welfare under individual rationality. In other words, they solve:
\begin{align}
\label{master_problem_in_text} 
    &\text{maximize}\quad W:\:(\collab,\bfn)\in\defEns{0,1}^J \times \R_{+}^J \mapsto \sum_{j\in\agents}\utility_{j} (\collab, \bfn)\\
    &\text{subject to}\quad
        \min_{j\in\agents}\utility_{j} ( \collab, \bfn)- \outside_j \geq 0\notag \eqsp.
\end{align}
In the Social Choice literature, $W$ is referred to as the utilitarian social  welfare function.
The participation constraint ensures that no agent within the coalition finds it beneficial to switch to their outside option. Note that $\bfn\in\R_{+}^J$ is required to have non-negative entries, which prevents the aggregator from giving away data points to agents. 

%%% Local Variables:
%%% mode: latex
%%% TeX-master: "main"
%%% End:

%
\section{Full-Information Benchmark: First-Best Collaboration}
\label{section:full_information}
In this section, we assume that the profile of types $(\theta_1,\ldots,\theta_J)\in\types^J$ is public, and study how the aggregator can implement an optimal collaboration among agents under this most-favorable scenario. 

\paragraph{Exact solution.}We are looking for a solution to the aggregator's problem \eqref{master_problem_in_text}. For any $\collab\in\{0,1\}^J$ and $\bfn \in \rset_+^J$, denote by%
\begin{align}
 \maxcontrib_j(\collab,\bfn) &=\max\defEnsLigne{n\geq0:\,u_j \parenthese{(1,\collab_{-j}), (n,\bfn_{-j})\sep\type_j} \geq \outside _j} \notag\\
  &=\noutside - (a/c)[\eps(\collab, \bfn)-\eps(\type_j, \noutside)]\label{eq:def_n_contrib_j}\eqsp,
\end{align}%
the maximum number of samples that agent $j$ can be asked to provide within the coalition under its participation constraint, where $\noutside$ is defined as in \eqref{eq:def_noutside}, and $\outside_j$ as in \eqref{eq:def_outside_j}. With this notation, problem \eqref{master_problem_in_text} rewrites
\begin{equation}
    \text{maximize}\quad W(\collab, \bfn)\quad\text{subject to}\quad\min_{j\in\cB}\maxcontrib_j (\collab, \bfn)- n_j \geq 0\label{master_problem_in_text2}\eqsp.
\end{equation}
\begin{restatable}{theorem}{optimalfullinfo}\label{prop:optimalfullinfo}
    Assume \Cref{assumption_upper_bound_risk}, \Cref{definition_type}, \Cref{assumption_types_sorted}. 
    Problem \eqref{master_problem_in_text} admits a unique solution  $(\Bopt, \bfnopt(\bfth))\in\defEns{0,1}^J \times \R_{+}^J$. Moreover,
    \begin{enumerate}[wide, label=(\roman*)]
        \item $\Bopt=\bOne=(1,\ldots,1)$,
        \item Denoting $\bfnopt(\bfth) = (\nopt_1(\bfth),\ldots,\nopt_J(\bfth))$, there exists $L^\opt\in\agents$ such that for any $j\in\agents$,
        \begin{equation}
\label{eq:optimalcontributionscheme}
\nopt_j (\bfth)\begin{cases}
    =\maxcontrib_j (\bOne,\bfnopt(\bfth)) &\text{if }j<L^\opt,\\
    \in [\,0\,,\,\maxcontrib_{j}(\bOne,\bfnopt(\bfth))\,]&\text{if }j=L^\opt,\\
    =0&\text{otherwise}.
\end{cases}
        \end{equation}
    \end{enumerate}
  \end{restatable}

The couple $(\Bopt, \bfnopt)$ is referred to as the \textit{optimal contribution scheme}. Although implicit, the condition \eqref{eq:optimalcontributionscheme} provides insights about the optimal scheme. The aggregator makes everyone enter the coalition, but only asks the $L^\mathrm{opt}>0$ first-best agents to contribute. This allows to obtain the best possible collaborative model while sparing any sampling cost to other agents. Moreover, the number of required samples $\nopt_j (\bfth)$ slightly differs from $\noutside$ according to the relative performance of the collaborative model with respect to agent $j$'s one: if the agent gets a better accuracy by collaborating, the aggregator can ask them for more data; if on the other hand the agent gets a worse model by collaborating (i.e., they are a contributor with very high quality data), the aggregator can only ask less data because of the participation constraint.

\paragraph{Relaxed solution.} Working with the optimal scheme $(\Bopt, \bfnopt)$ is difficult because $\maxcontrib_j (\bOne, \bfnopt(\bfth))$ has no explicit expression. To make the analysis tractable, we slightly simply the optimal contribution scheme in \Cref{prop:optimalfullinfo} and consider the \text{simplified optimal contribution scheme} $(\collab^\star, \bfns)$ where $\Bstar=\Bopt=(1,\ldots,1)$ and for any $j\in\agents$,
    \begin{align}
        \nstar_j (\bfth)=\2{j\leq\Lstar} \maxcontrib_j (\bOne, \bfns(\bfth))\quad &\text{and}\quad\Lstar=\min\defEnsLigne{j\in\agents:\: \sum_{k\leq j}\maxcontrib_k (\bOne, \bfns(\bfth))\geq\Nbar}\label{def:simplifiedscheme}\eqsp,\\
        &\text{with}\quad\Nbar=(\noutside+1)J^{\frac{1}{1+\gamma}}-1\eqsp.\notag
    \end{align}
   Note that $(\Bstar,\bfns)$ only differs from $(\Bopt, \bfnopt)$ in two ways. First, in $(\Bstar, \bfns)$ all contributors'  participation constraint bind, while in $(\Bopt, \bfnopt)$ the $L^\opt$-the one could be slack. Second, the total number of data points required from the coalition is fixed and equal to $\Nbar=\Theta(J^{\frac{1}{1+\gamma}})$. The quantity $\Nbar$ comes from a natural relaxation of the original problem \Cref{master_problem_in_text2} where we leave aside an intricate term of the objective function. This relaxation, which is formally described in \Cref{appendix:relaxation},  provides a good approximation of the exact solution in reasonable settings. Indeed, the following result establishes that applying $(\Bstar, \bfns)$ instead of $(\Bopt, \bfnopt)$ comes at a negligible welfare cost when types are sufficiently evenly spaced.
   \begin{restatable}{lemma}{simplifiednegligiblewelfareloss}\label{lemma:approxwelfare}Assume \Cref{assumption_upper_bound_risk}, \Cref{definition_type} and $\type_j - \type_{j-1}=\bigO(1/J)$ for any $j\in\defEnsLigne{2,\ldots,J}$. Then, $$
W(\Bopt, \bfnopt(\bfth))=W(\collab^\star, \bfns(\bfth))+\bigO(J^{\frac{1}{1+\gamma}})\eqsp.$$
\end{restatable}
   
   Moreover, the following proposition shows that  $\bfns(\bfth)$ admits a workable expression.
   
\begin{restatable}{corollary}{simplifyingfullinfo}\label{cor:simplifyingfullinfo}
        Assume \Cref{assumption_upper_bound_risk}, \Cref{definition_type} and  \Cref{assumption_types_sorted}. Then $\Lstar=\Theta (J^{\frac{1}{1+\gamma}})$ and for any $j\in\agents$,
$$
\nstar_{j}(\bfth) = \2{j\le \Lstar}\parentheseDeux{\frac{\Nbar}{\Lstar} + \frac{2a}{c}\parenthese{\type_j - \frac{1}{\Lstar}\sum_{\ell=1}^{\Lstar} \type_\ell}}\eqsp.
$$\end{restatable}

% \Aymeric{Pas sûr de vouloir garder ce résultat dans le texte, il n'est jamais ré-utilisé et est rigolo mais sans plus}
% In addition, we obtain an explicit dependency of the collaborative model risk excess with respect to the first and second order conditional moments of the type distribution:
% $$\riskexcess(\bOne, \bfns(\bfth))=\alpha_{\delta}(1+\Nbar)^{-\gamma} + \beta + 2\empmean{\Lstar} + (4a\Lstar/c\Nbar)\empvar{\Lstar}\eqsp,$$
% with $\empvar{\Lstar}=(1/\Lstar)\sum_{j=1}^{\Lstar} (\type_j - \empmean{\Lstar})^2$.
% \end{restatable}
 Since $(\collab^\star, \bfns)$ correctly approximates the optimal scheme while being more tractable, we work with it in the remainder to lighten proofs.
\begin{assumption}\label{assumption:simplifiedscheme}
    The aggregator applies the simplified contribution scheme $(\collab^\star, \bfns)$.
\end{assumption}
%%% Local Variables:
%%% mode: latex
%%% TeX-master: "main"
%%% End:

%
\section{Hidden information}\label{section:hidden_information}The welfare-maximizing contribution scheme described in \eqref{def:simplifiedscheme} depends explicitly on $\bfth\in\types^J$, so it is implementable only if types are public. This often unrealistic, either for legal or competitive reasons. We therefore turn to the problem of setting up a collaboration when types are private.
\subsection{Naive aggregation and unravelling} \label{subsection_naive_aggregation}

A naive solution to coping with the private nature of $\bfth\in\types^J$ is for the aggregator to ask agents to disclose their types, and apply the simplified optimal contribution scheme defined in \eqref{def:simplifiedscheme}. In this setting, however, agents may declare a type $\tilde{\theta}_j$ different from their true type $\theta_j$. 

This approach corresponds to a direct-revelation mechanism $\Gamma:(\collab,\tbfth)\mapsto \bfns(\tbfth)$ which unfolds as follows.
\begin{enumerate}
    \item Any agent $j\in\agents$ declares a tuple $(B_j, \ttype_j)\in\defEns{0,1}\times\types\cup\defEns{\dagger}$. If $B_j = 1$, then agent $j$ picks \ref{item:option2}, and enters the coalition with type $\ttype_j$. If $B_j=0$, then agent $j$ picks \ref{item:option1}, their declared type $\ttype_j$ is $\dagger$ by convention.
 
    \item Setting $\msb = (B_1,\ldots,B_J)$ and $\tilde{\boldsymbol{\theta}} \in (\Theta \cup \{\dagger\})^J$,  then the aggregator applies the contribution scheme defined in \eqref{def:simplifiedscheme}, so the vector of number of contributions within the coalition is $\bfns (\tbfth)$. 
\end{enumerate}
$\Gamma$ induces a direct revelation game $\parentheseLigne{\agents, \strategyspace ^J, \parentheseLigne{v_j}_{j\in\agents}}$ where the action space is $\strategyspace = \ensembleLigne{(1,\ttype)}{\ttype\in\types}\union\defEns{(0,\dagger)}$ and payoffs are for any $j\in\agents$ and $\bfs\in\strategyspace^J$,
\begin{equation*}
    v_j(s_j, \bfs_{-j})=\utility_{j} (\collab, \bfns(\tbfth))= B_j\parentheseDeux{-a\parenthese{\Rstar + \riskexcess(\collab, \bfns(\tbfth))}-c\nstar_j(\tbfth)}+(1-B_j)\outside_j \eqsp.
\end{equation*}
This mechanism is obviously vulnerable to strategic manipulation, since it disregards incentive compatibility. This has severe consequences for the coalition, as shown by the following proposition.

\begin{restatable}[Unravelling]{theorem}{unravelling}\label{prop:unravelling}Assume \Cref{assumption_upper_bound_risk}, \Cref{definition_type}, \Cref{assumption_types_sorted},   and \Cref{assumption:simplifiedscheme}. Let $\nasheqspace\subset\strategyspace^J$ be the set of pure-strategy Nash equilibria of the game induced by $\Gamma$. We have\begin{enumerate}[label=(\roman*)]
    \item $\nasheqspace \ne \emptyset$
    \item at any
    $\sstar\in\nasheqspace$, $\collab=(0,\ldots,0)$ or
    $\collab=(0,\ldots,0,1)\eqsp.$
\end{enumerate} \end{restatable}
\Cref{prop:unravelling} shows that under $\Gamma$, the coalition undergoes a full unravelling:  it is either empty or comprised solely of the worst agent in any Nash equilibrium. Thus, collaborative learning is not immune to adverse selection, and may suffer from unravelling as any market characterized by information asymmetry.
\begin{proof}[Sketch of proof]
The profile of actions $((0,\dagger),\ldots,(0,\dagger))$ corresponding to $\collab=(0,\ldots,0)$ is a pure Nash equilibrium, since forming a lone coalition cannot bring more utility than picking the outside option. Conversely, consider a pure-strategy Nash equilibrium $\bs\in\cE$ such that $\collab\ne(0,\ldots,0)$. Denote by $\cC=\defEnsLigne{j\in\agents:\:B_j = 1\text{   and  
   } \nstar_j(\tbfth)>0}$ the set of  contributors under this equilibrium. It can be shown that (i) for any $(j,k)\in \cC^2$, $\type_j - \ttype_j = \type_k - \ttype_k$, and (ii) For any $j\in \cC$, $(1,\ttype_j)$ with $\ttype_j > \typemin$ is strictly dominated by  $(1,\typemin)$ so $\ttype_j = \typemin$ at the equilibrium. As a consequence, $\type_j = \type_k$ for any $(j,k)\in \cC^2$, which implies by \Cref{assumption_types_sorted} that $\abs{\cC}=1$. From the definition of the contribution scheme \eqref{def:simplifiedscheme}, we can deduce that $\sum_{j\in\agents}B_j = \abs{\cB}=1$, because $\abs{\cB}>1$ would entail $\abs{\cC}>1$. Finally, $B_J = 1$ because $v_J ((1,\typemin),\bs _{-J}) > v_J ((0,\dagger),\bs _{-J})$ and $\bs$ is a Nash equilibrium. This leads to $\collab = (0,\ldots,0,1)$.
\end{proof}

\subsection{Breaking unravelling}
The previous results motivate the design of a more sophisticated aggregation scheme that addresses adverse selection. In this section, we discuss how to design such a procedure. 
\paragraph{Is VCG available in our framework?}Unravelling occurs under $\Gamma$ because agents do not find it beneficial to declare their true type and eventually opt for their outside option. This could be avoided by modifying $\Gamma$ to make it
\begin{enumerate}[label=(\roman*)]
    \item \textit{individually rational}, that is $\text{$v_j((1,\type_j),\bfs_{-j})\geq v_j ((0,\dagger),\bfs_{-j})$ for any $j\in\agents$, $\bfs_{-j}\in\strategyspace^{J-1}$}$,
    \item and \textit{incentive compatible}, that is $\text{$v_j ( (1,\type_j), \bfs_{-j})\geq v_j ((1,\ttype_j),\bfs_{-j})$ for any $\ttype_j\in\types$}$.
\end{enumerate}
Under these conditions, the truthful, optimal profile of actions $((1,\type_1),\ldots,(1,\type_J))$ would emerge as a Nash equilibrium. 
Since the aggregator seeks to minimize the utilitarian function $W$, one option would be to rely on the VCG mechanism ~\citep{v61, c71, g73}, which is the direct-revelation mechanisms fulfilling these desiderata \citep{gl77, h79}. Formally, the VCG mechanism writes $\Gamma^{\textsc{vcg}}:\tbfth \mapsto (\bfns(\tbfth), \mathbf{t}(\tbfth))$ where $\mathbf{t}(\tbfth)=(t_1 (\bfth),\ldots, t_J (\bfth))\in\R^J$ is a set of transfers satisfying for any $j\in\agents$:\begin{equation*}
    t_j(\tbfth) = \sum_{k\ne j}u_k((0,\bOne_{-j}), \bfns(\tbfth))-\sum_{k\ne j}u_k (\bOne, \bfns(\tbfth))\eqsp.
\end{equation*}It re-establishes truthfulness as a dominant strategy by aligning individual payoffs $v_j ^{\textsc{VCG}}(\tbfth)=u_j (\bOne, \bfns(\tbfth))-t_j (\tbfth)$ with total social welfare. Unfortunately, the VCG approach is unavailable in our framework, because of the following observation.
\begin{restatable}{lemma}{positivepaymentvcg}\label{lemma:positivepaymentvcg}
    There exists $j\in\agents$ such that $-t_j (\tbfth)>0$.
\end{restatable}
\Cref{lemma:positivepaymentvcg} shows that some agent would need to receive a strictly positive transfer. This is impossible without a monetary payment--utility can only be decreased by the aggregator, for instance through \textit{accuracy shaping} \citep{kwj22}--, which we exclude here.

\paragraph{A probabilistic verification-based mechanism.}We now show how to design a mechanism that recovers the optimal collaboration as a Nash equilibrium in high probability without the need for transfers. Inspired by the probabilistic verification approach \citep{cie12, fv18, bk19}, we assume that the aggregator can approximately estimate $\type_j$ with few samples from $P_j$ for any $j\in\agents$:

\begin{assumption}
    \label{assumption:typeestimator}There exists a decreasing function $\eta_{\delta}:\: \R_+ ^\star \rightarrow \R_+ ^\star$, with $\delta\in(0,1)$ defined in \Cref{assumption_upper_bound_risk}, such that for any $j\in\agents$ and i.i.d samples $(X^j _1, Y^j _1),\ldots, (X^j _q, Y^j _q)$ of size $q>0$ from $P_j$, there exists a $(X^j _1, Y^j _1),\ldots, (X^j _q, Y^j _q)$-measurable estimator $\htype_j$  satisfying
    $$\P\parentheseLigne{\absLigne{\widehat{\type}_j - \type_j}\leq \eta_{\delta}(q)}\geq 1-\delta\eqsp.$$\end{assumption}In \Cref{para:statisticalimplementation}, we show how the aggregator can design such estimators. \Cref{assumption:typeestimator} allows us to consider a new mechanism $\widehat{\Gamma}:\:\collab\mapsto\bfm(\collab)$ as follows:
\begin{enumerate}
    \item any agent $j\in\agents$ declares $B_j \in\defEnsLigne{0,1}$. If $B_j = 1$, the principal asks for $\underline{q} \leq \noutside - 2(a/c)(\typemax-\typemin)$ i.i.d samples from $P_j$ and estimates types as $\hbfth=(\htype_j)_{j\in\cB}$ following \Cref{assumption:typeestimator}.
    \item Based on the estimated types $\hbfth=(\htype_j)_{j\in\cB}$, the aggregator asks for $\parentheseDeuxLigne{\,\nstar_j (\hbfth+\bfeta_j)-\qubar\,}_{+}$ additional samples from $P_j$, where $\bfns(\centraldot)$ is defined as in \eqref{def:simplifiedscheme} and
    $$
       \bfeta_j = \etadl\bOne - 2\boldsymbol{\delta}_j \etadl , \quad\text{with}\quad \boldsymbol{\delta}_j = (0,\ldots,0,1,0,\ldots,0)^{\transpose} \eqsp.
    $$
    
    Thus, the number of draws required from agent $j$  is $\max[\qubar\,,\,\nstar_j(\hbfth + \bfeta_j)]$.
    \item The aggregator keeps
$m_j (\hbfth) = \1\defEnsLigne{\nstar_j (\hbfth + \bfeta_j)>0}\max[\qubar\,,\,\nstar_j(\hbfth + \bfeta_j)]$ samples from agent $j$, and trains a collaborative model with these pooled samples.
    \end{enumerate} 

$\hat{\Gamma}$ induces a game $\parentheseLigne{\agents, \defEnsLigne{0,1} ^J, \parentheseLigne{\hv_j}_{j\in\agents}}$ where any agent $j\in\agents$ has a payoff function
\begin{equation*}
    \hv_j: (B_j, \collab_{-j})\mapsto  B_j\parentheseDeux{-a\parenthese{\Rstar + \riskexcess(\collab, \bfm(\hbfth))}-c\max[\qubar\,,\,\nstar_j(\hbfth + \bfeta_j)]}+(1-B_j)\outside_j \eqsp.
\end{equation*}

The rationale behind this mechanism is fairly intuitive: since $\hbfth$ is a correct estimate of $\bfth$, $ \nstar_j(\hbfth + \bfeta_j)\approx \nstar_j (\hbfth)$ correctly approximates the optimal contribution $\nstar_j (\bfth)$ for any contributor~$j\in\cB$. Note that type estimates are purposely biased by $\bfeta_j$ when asking for samples. This is a safeguard against over-estimated types, which would lead to asking to many data points and could deter agents from participating in the coalition.

Critically, $\bfm(\hbfth)$ does not depend on declared type, so agents are no longer able to strategically manipulate the mechanism. Moreover, $\hat{\Gamma}$ does not require agents to know their own types, which would be an unrealistic assumption. Finally, observe that the number of data points asked to estimate types $\qubar$ is low enough to never deter agents from participating in the coalition.
% Asking for less than $\qubar$ data points is possible and would entail a lesser loss of welfare, but would also result in noisier type estimates $\hbfth$. This in turn would lead to a greater Nash equilibrium approximation term in the following result. The fine optimization of the $\qubar$ parameter to optimally address this trade-off is left for future work. 

\begin{restatable}{theorem}{truthfulnessisnash}
\label{theorem:truthfulnessisnash}Assume \Cref{assumption_upper_bound_risk}, \Cref{definition_type},  \Cref{assumption_types_sorted},  \Cref{assumption:simplifiedscheme} and    \Cref{assumption:typeestimator}. $\Bstar=(1,\ldots,1)^{\transpose}$ is a Nash equilibrium under $\hat{\Gamma}$ with probability $1-\delta$.
\end{restatable}
\Cref{theorem:truthfulnessisnash} shows that the optimal  coalition is a sustainable equilibrium under $\hat{\Gamma}$, which effectively breaks unravelling: the set of (approximate) Nash equilibria is no more reduced to profiles of actions where the coalition is empty, or reduced to the worst agent. 
\paragraph{Practical implementation.}\label{para:statisticalimplementation}We now explain how to practically implement $\widehat{\Gamma}$ in a collaborative learning setting. This requires defining a collection of estimators $(\htype_j)_j$ satisfying \Cref{assumption:typeestimator}. To this end, we assume that few  samples from the target distribution  are available.
\begin{assumption}
    \label{assumption_sample_from_Pzero}
    There are $\qprime >0$ i.i.d samples $\defEns{(X^0 _1 , Y^0 _1),\ldots, (X^0 _{\qprime} , Y^0 _{\qprime})}$ from $P_0$ available to the  aggregator and agents.
  \end{assumption}Under \Cref{assumption_sample_from_Pzero}, define $\emprisk{g}{0} = q^{\prime -1}\sum_{i=1}^{\qprime}\ell(g(X^0 _i), Y^0 _i)$ for $g\in\hypoclass$. 
This allows us to devise suitable estimators $\htype_j$ as follows. 
\begin{restatable}{proposition}{estimatorissuitable}\label{prop:estimatorissuitable}
  Assume \Cref{assumption_upper_bound_risk}, \Cref{definition_type} and \Cref{assumption_sample_from_Pzero}. For any $j\in\agents$ the estimator $$\htype^{\textsc{erm}}_{0,j}=\sup_{g\in\hypoclass}\absLigne{\emprisk{g}{j}-\emprisk{g}{0}}\eqsp,$$ satisfies  \Cref{assumption:typeestimator} with \begin{equation}\label{eq:valueta}
      \eta_{\delta}(q) = \alpha_{\delta/4}\parentheseDeux{(q+1)^{-\gamma}+(\qprime+1)^{-\gamma}} + 2\beta\eqsp.
  \end{equation}  
\end{restatable}
\Cref{prop:estimatorissuitable} shows that the empirical version of the $\hypoclass$-divergence defined in \Cref{definition_type} correctly estimate types. Note that the tighter the PAC bound in \Cref{assumption_upper_bound_risk}, the better the approximation term in \eqref{eq:valueta}. The type estimator $\htype_{0,j} ^{\textsc{ERM}}$ defined in \Cref{assumption:typeestimator} can easily be computed in  the classification case, as shown with the following example. 
\begin{assumption}[Classification setting]\label{assumption:classification}
$\cY=\defEns{-1,1}$, $\ell=\ell_{0,1}:(y,y')\in\cY\times\cY\mapsto \2{yy'<0}$, and $\hypoclass$ is a symmetric ($g\in\hypoclass$ if and only if $-g\in\hypoclass$) class of classifiers.
\end{assumption}

\begin{restatable}{example}{exampleclassif}\label{example:classif}Assume  \Cref{assumption_upper_bound_risk}, \Cref{definition_type}, \Cref{assumption_sample_from_Pzero} and \Cref{assumption:classification}.
\begin{enumerate}[wide, labelwidth=!, labelindent=0pt, label=(\roman*)]
    \item Denoting $\emprisk{g}{j^{-}}= n_j ^{-1}\sum_{i=1}^{n_j}\ell_{0,1}(g(X^j _i), -Y^j _i)$, we have $$
    \htype^{\textsc{erm}}_{0,j} = 1-\inf_{g\in\hypoclass}\defEns{\emprisk{g}{0}+\emprisk{g}{j^{-}}}\eqsp.$$
    \item In \Cref{assumption_upper_bound_risk}, assume $\alpha_{\delta}=\ln(1/\delta)^{1/2}$, $\beta=2\textsc{RAD}(\hypoclass)$ and $\gamma=1$ \citep{b04}. With $\htype_{0,j} ^{\textsc{erm}}$ defined in \Cref{prop:estimatorissuitable}, we have \begin{align*}\eta_{\delta/J}(q)=\ln(4J / \delta)^{1/2}[(1+q)^{-\gamma}+(1+\qprime)^{-\gamma}] + 2\mathrm{Rad}(\hypoclass)\eqsp.\end{align*}
\end{enumerate}
\end{restatable}
\Cref{example:classif} shows that in the classification case, it is sufficient to flip the labels of the data received from each contributor, merge these samples with those from $P_0$, and perform an empirical risk minimization to compute $\htype_{0,j}^\textsc{erm}$. The approximation error grows no more than logarithmically with the number of agents, while decreasing at rate $\gamma$ with the number of samples used in the estimation.
%%% Local Variables:

%
\section{Conclusion}
In this work, we show that information asymmetry has deleterious consequences when strategic agents try to learn a collaborative model. More precisely, under a naive aggregation procedure, the ignorance of others' data quality leads the coalition of learners to be either empty or reduced to the lowest-quality agent. We introduce a transfer-free mechanism based on estimation of types. This effectively counteracts unravelling by letting the grand coalition ranks among the approximate Nash equilibria with high probability. 

Several possible extensions can be considered. First, it would be interesting to relax the assumption that all agents have the same ratio $a/c$ in their utility, and see how heterogeneity affects the results. Second, the mechanism presented in \Cref{section:hidden_information} aims for individual rationality. A more desirable, yet difficult to achieve, property would be core stability, to ensure that no group of agents would benefit from a coordinated deviation, i.e., forming an alternative coalition. Finally, it would be interesting to check whether there exist mechanisms where the optimal collaboration not only emerges as a Nash equilibrium, but as a dominant equilibrium under imperfect information.

\section*{Acknowledgements}
Funded by the European Union (ERC, Ocean, 101071601). Views and opinions expressed are however those of the author(s) only and do not necessarily reflect those of the European Union or the European Research Council Executive Agency. Neither the European Union nor the granting authority can be held responsible for them.

\bibliographystyle{plainnat}
\bibliography{sample}

\begin{thebibliography}{38}
\providecommand{\natexlab}[1]{#1}
\providecommand{\url}[1]{\texttt{#1}}
\expandafter\ifx\csname urlstyle\endcsname\relax
  \providecommand{\doi}[1]{doi: #1}\else
  \providecommand{\doi}{doi: \begingroup \urlstyle{rm}\Url}\fi

\bibitem[Akerlof(1970)]{a70}
George~A. Akerlof.
\newblock {The market for lemons: Quality uncertainty and the market mechanism}.
\newblock \emph{The Quarterly Journal of Economics}, 84\penalty0 (3):\penalty0 488--500, 1970.
\newblock URL \url{https://ideas.repec.org/a/oup/qjecon/v84y1970i3p488-500..html}.

\bibitem[Ananthakrishnan et~al.(2023)Ananthakrishnan, Bates, Jordan, and Haghtalab]{abjh23}
Nivasini Ananthakrishnan, Stephen Bates, Michael Jordan, and Nika Haghtalab.
\newblock Delegating data collection in decentralized machine learning, 2023.

\bibitem[Ball and Kattwinkel(2019)]{bk19}
Ian Ball and Deniz Kattwinkel.
\newblock {Probabilistic verification in mechanism design}.
\newblock September 2019.

\bibitem[Ben-David et~al.(2010)Ben-David, Blitzer, Crammer, Kulesza, Pereira, and Vaughan]{bsb10}
Shai Ben-David, John Blitzer, Koby Crammer, Alex Kulesza, Fernando Pereira, and Jennifer Vaughan.
\newblock A theory of learning from different domains.
\newblock \emph{Machine Learning}, 79:\penalty0 151--175, 05 2010.
\newblock \doi{10.1007/s10994-009-5152-4}.

\bibitem[Blum et~al.(2017)Blum, Haghtalab, Procaccia, and Qiao]{blum-etal}
Avrim Blum, Nika Haghtalab, Ariel Procaccia, and Mingda Qiao.
\newblock Collaborative {PAC} learning.
\newblock In I.~Guyon, U.~Von Luxburg, S.~Bengio, H.~Wallach, R.~Fergus, S.~Vishwanathan, and R.~Garnett, editors, \emph{Advances in Neural Information Processing Systems}. Curran Associates, 2017.

\bibitem[Blum et~al.(2021)Blum, Haghtalab, Phillips, and Shao]{bhp21}
Avrim Blum, Nika Haghtalab, Richard~Lanas Phillips, and Han Shao.
\newblock One for one, or all for all: Equilibria and optimality of collaboration in federated learning.
\newblock In Marina Meila and Tong Zhang, editors, \emph{Proceedings of the 38th International Conference on Machine Learning}, volume 139 of \emph{Proceedings of Machine Learning Research}, pages 1005--1014. PMLR, 18--24 Jul 2021.
\newblock URL \url{https://proceedings.mlr.press/v139/blum21a.html}.

\bibitem[Bousquet et~al.(2003)Bousquet, Boucheron, Lugosi, Luxburg, and Rätsch]{b04}
Olivier Bousquet, Stéphane Boucheron, Gábor Lugosi, Ulrike Luxburg, and Gunnar Rätsch.
\newblock Introduction to statistical learning theory.
\newblock \emph{Advanced Lectures on Machine Learning, 169-207 (2004)}, 01 2003.
\newblock \doi{10.1007/978-3-540-28650-9_8}.

\bibitem[Caragiannis et~al.(2012)Caragiannis, Elkind, Szegedy, and Yu]{cie12}
Ioannis Caragiannis, Edith Elkind, Mario Szegedy, and Lan Yu.
\newblock Mechanism design: From partial to probabilistic verification.
\newblock \emph{Proceedings of the ACM Conference on Electronic Commerce}, 06 2012.
\newblock \doi{10.1145/2229012.2229035}.

\bibitem[Clarke(1971)]{c71}
Edward Clarke.
\newblock {Multipart pricing of public goods}.
\newblock \emph{Public Choice}, 11\penalty0 (1):\penalty0 17--33, September 1971.
\newblock \doi{10.1007/BF01726210}.
\newblock URL \url{https://ideas.repec.org/a/kap/pubcho/v11y1971i1p17-33.html}.

\bibitem[Donahue and Kleinberg(2021{\natexlab{a}})]{kk21}
Kate Donahue and Jon Kleinberg.
\newblock Optimality and stability in federated learning: A game-theoretic approach, 06 2021{\natexlab{a}}.

\bibitem[Donahue and Kleinberg(2021{\natexlab{b}})]{kk212}
Kate Donahue and Jon Kleinberg.
\newblock Model-sharing games: Analyzing federated learning under voluntary participation.
\newblock \emph{Proceedings of the AAAI Conference on Artificial Intelligence}, 35:\penalty0 5303--5311, 05 2021{\natexlab{b}}.
\newblock \doi{10.1609/aaai.v35i6.16669}.

\bibitem[Dorner et~al.(2023)Dorner, Konstantinov, Pashaliev, and Vechev]{dk23}
Florian~E. Dorner, Nikola Konstantinov, Georgi Pashaliev, and Martin~T. Vechev.
\newblock Incentivizing honesty among competitors in collaborative learning and optimization.
\newblock \emph{ArXiv}, abs/2305.16272, 2023.
\newblock URL \url{https://api.semanticscholar.org/CorpusID:258887502}.

\bibitem[Einav and Finkelstein(2011)]{ef11}
Liran Einav and Amy Finkelstein.
\newblock Selection in insurance markets: Theory and empirics in pictures.
\newblock \emph{Journal of Economic Perspectives}, 25\penalty0 (1):\penalty0 115--38, March 2011.
\newblock \doi{10.1257/jep.25.1.115}.
\newblock URL \url{https://www.aeaweb.org/articles?id=10.1257/jep.25.1.115}.

\bibitem[Ferraioli and Ventre(2018)]{fv18}
Diodato Ferraioli and Carmine Ventre.
\newblock Probabilistic verification for obviously strategyproof mechanisms, 2018.

\bibitem[Fu et~al.(2023)Fu, Zhang, Gao, Zhang, and Liu]{fu2023client}
Lei Fu, Huanle Zhang, Ge~Gao, Mi~Zhang, and Xin Liu.
\newblock Client selection in federated learning: Principles, challenges, and opportunities, 2023.

\bibitem[Gao et~al.(2022)Gao, Yao, and Yang]{gao2022survey}
Dashan Gao, Xin Yao, and Qiang Yang.
\newblock A survey on heterogeneous federated learning.
\newblock \emph{arXiv preprint arXiv:2210.04505}, 2022.

\bibitem[Green and Laffont(1977)]{gl77}
Jerry Green and Jean-Jacques Laffont.
\newblock Characterization of satisfactory mechanisms for the revelation of preferences for public goods.
\newblock \emph{Econometrica}, 45\penalty0 (2):\penalty0 427--38, 1977.
\newblock URL \url{https://EconPapers.repec.org/RePEc:ecm:emetrp:v:45:y:1977:i:2:p:427-38}.

\bibitem[Groves(1973)]{g73}
Theodore Groves.
\newblock Incentives in teams.
\newblock \emph{Econometrica}, 41:\penalty0 617--631, 1973.
\newblock URL \url{https://api.semanticscholar.org/CorpusID:264740987}.

\bibitem[Hendren(2013)]{h13}
Nathaniel Hendren.
\newblock {Private Information and Insurance Rejections}.
\newblock \emph{Econometrica}, 81\penalty0 (5):\penalty0 1713--1762, September 2013.
\newblock \doi{ECTA10931}.
\newblock URL \url{https://ideas.repec.org/a/ecm/emetrp/v81y2013i5p1713-1762.html}.

\bibitem[Holmstrom(1979)]{h79}
Bengt Holmstrom.
\newblock Groves' scheme on restricted domains.
\newblock \emph{Econometrica}, 47\penalty0 (5):\penalty0 1137--44, 1979.
\newblock URL \url{https://EconPapers.repec.org/RePEc:ecm:emetrp:v:47:y:1979:i:5:p:1137-44}.

\bibitem[Huang et~al.(2023)Huang, Karimireddy, and Jordan]{hkm23}
Baihe Huang, Sai~Praneeth Karimireddy, and Michael~I. Jordan.
\newblock {Evaluating and Incentivizing Diverse Data Contributions in Collaborative Learning}.
\newblock Papers 2306.05592, arXiv.org, June 2023.
\newblock URL \url{https://ideas.repec.org/p/arx/papers/2306.05592.html}.

\bibitem[Huang et~al.(2022)Huang, Ke, Kamhoua, Mohapatra, and Liu]{hk22}
Chao Huang, Shuqi Ke, Charles Kamhoua, Prasant Mohapatra, and Xin Liu.
\newblock Incentivizing data contribution in cross-silo federated learning, 2022.

\bibitem[Kairouz et~al.(2021)Kairouz, McMahan, Avent, Bellet, Bennis, Bhagoji, Bonawitz, Charles, Cormode, Cummings, D'Oliveira, Eichner, Rouayheb, Evans, Gardner, Garrett, Gascón, Ghazi, Gibbons, Gruteser, Harchaoui, He, He, Huo, Hutchinson, Hsu, Jaggi, Javidi, Joshi, Khodak, Konečný, Korolova, Koushanfar, Koyejo, Lepoint, Liu, Mittal, Mohri, Nock, Özgür, Pagh, Raykova, Qi, Ramage, Raskar, Song, Song, Stich, Sun, Suresh, Tramèr, Vepakomma, Wang, Xiong, Xu, Yang, Yu, Yu, and Zhao]{kairouz2021advances}
Peter Kairouz, H.~Brendan McMahan, Brendan Avent, Aurélien Bellet, Mehdi Bennis, Arjun~Nitin Bhagoji, Kallista Bonawitz, Zachary Charles, Graham Cormode, Rachel Cummings, Rafael G.~L. D'Oliveira, Hubert Eichner, Salim~El Rouayheb, David Evans, Josh Gardner, Zachary Garrett, Adrià Gascón, Badih Ghazi, Phillip~B. Gibbons, Marco Gruteser, Zaid Harchaoui, Chaoyang He, Lie He, Zhouyuan Huo, Ben Hutchinson, Justin Hsu, Martin Jaggi, Tara Javidi, Gauri Joshi, Mikhail Khodak, Jakub Konečný, Aleksandra Korolova, Farinaz Koushanfar, Sanmi Koyejo, Tancrède Lepoint, Yang Liu, Prateek Mittal, Mehryar Mohri, Richard Nock, Ayfer Özgür, Rasmus Pagh, Mariana Raykova, Hang Qi, Daniel Ramage, Ramesh Raskar, Dawn Song, Weikang Song, Sebastian~U. Stich, Ziteng Sun, Ananda~Theertha Suresh, Florian Tramèr, Praneeth Vepakomma, Jianyu Wang, Li~Xiong, Zheng Xu, Qiang Yang, Felix~X. Yu, Han Yu, and Sen Zhao.
\newblock Advances and open problems in federated learning, 2021.

\bibitem[Karimireddy et~al.(2022)Karimireddy, Guo, and Jordan]{kwj22}
Sai~Praneeth Karimireddy, Wenshuo Guo, and Michael~I. Jordan.
\newblock Mechanisms that incentivize data sharing in federated learning, 2022.

\bibitem[Kifer et~al.(2004)Kifer, Ben-David, and Gehrke]{kb04}
Daniel Kifer, Shai Ben-David, and Johannes Gehrke.
\newblock Detecting change in data streams.
\newblock pages 180--191, 04 2004.

\bibitem[Konstantinov and Lampert(2019)]{k19}
Nikola Konstantinov and Christoph Lampert.
\newblock Robust learning from untrusted sources, 2019.

\bibitem[Laffont and Martimort(2001)]{lm01}
Jean-Jacques Laffont and David Martimort.
\newblock \emph{The Theory of Incentives: The Principal-Agent Model}.
\newblock Princeton University Press, Princeton, NJ, USA, 2001.

\bibitem[Liu et~al.(2023)Liu, Li, and Zhu]{lst23}
Shutian Liu, Tao Li, and Quanyan Zhu.
\newblock Game-theoretic distributed empirical risk minimization with strategic network design.
\newblock \emph{IEEE Transactions on Signal and Information Processing over Networks}, 9:\penalty0 542--556, 2023.
\newblock \doi{10.1109/TSIPN.2023.3306106}.

\bibitem[Mas-Colell et~al.(1995)Mas-Colell, Whinston, and Green]{mc95}
Andreu Mas-Colell, Michael~D. Whinston, and Jerry~R. Green.
\newblock \emph{{Microeconomic Theory}}.
\newblock Number 9780195102680 in OUP Catalogue. Oxford University Press, 1995.
\newblock ISBN ARRAY(0x516030e0).
\newblock URL \url{https://ideas.repec.org/b/oxp/obooks/9780195102680.html}.

\bibitem[Rothschild and Stiglitz(1976)]{rs76}
Michael Rothschild and Joseph Stiglitz.
\newblock Equilibrium in competitive insurance markets: An essay on the economics of imperfect information.
\newblock \emph{The Quarterly Journal of Economics}, 90\penalty0 (4):\penalty0 629--649, 1976.
\newblock URL \url{https://EconPapers.repec.org/RePEc:oup:qjecon:v:90:y:1976:i:4:p:629-649.}

\bibitem[Saig et~al.(2023)Saig, Talgam-Cohen, and Rosenfeld]{saig2023delegated}
Eden Saig, Inbal Talgam-Cohen, and Nir Rosenfeld.
\newblock Delegated classification, 2023.

\bibitem[Shalaeva et~al.(2019)Shalaeva, Esfahani, Germain, and Petreczky]{sfp19}
Vera Shalaeva, Alireza~Fakhrizadeh Esfahani, Pascal Germain, and Mih{\'a}ly Petreczky.
\newblock Improved pac-bayesian bounds for linear regression.
\newblock In \emph{AAAI Conference on Artificial Intelligence}, 2019.
\newblock URL \url{https://api.semanticscholar.org/CorpusID:208857400}.

\bibitem[Tsoy and Konstantinov(2024)]{toik23}
Nikita Tsoy and Nikola Konstantinov.
\newblock Strategic data sharing between competitors.
\newblock \emph{Advances in Neural Information Processing Systems}, 36, 2024.

\bibitem[Tu et~al.(2022)Tu, Zhu, Luong, Niyato, Zhang, and Li]{xkn22}
Xuezhen Tu, Kun Zhu, Nguyen~Cong Luong, Dusit Niyato, Yang Zhang, and Juan Li.
\newblock Incentive mechanisms for federated learning: From economic and game theoretic perspective.
\newblock \emph{IEEE Transactions on Cognitive Communications and Networking}, 8\penalty0 (3):\penalty0 1566--1593, 2022.
\newblock \doi{10.1109/TCCN.2022.3177522}.

\bibitem[Vickrey(1961)]{v61}
William Vickrey.
\newblock Counterspeculation, auctions, and competitive sealed tenders.
\newblock \emph{Journal of Finance}, 16\penalty0 (1):\penalty0 8--37, 1961.
\newblock URL \url{https://EconPapers.repec.org/RePEc:bla:jfinan:v:16:y:1961:i:1:p:8-37}.

\bibitem[Wei et~al.(2024)Wei, Li, Ren, Xu, and Wang]{wl24}
Zhepei Wei, Chuanhao Li, Tianze Ren, Haifeng Xu, and Hongning Wang.
\newblock Incentivized truthful communication for federated bandits, 2024.

\bibitem[Werner et~al.(2024)Werner, Karimireddy, and Jordan]{w24}
Mariel Werner, Sai~Praneeth Karimireddy, and Michael~I. Jordan.
\newblock Defection-free collaboration between competitors in a learning system, 2024.
\newblock URL \url{https://arxiv.org/abs/2406.15898}.

\bibitem[Yan et~al.(2023)Yan, Tang, Huang, and Tang]{yx23}
Yizhou Yan, Xinyu Tang, Chao Huang, and Ming Tang.
\newblock Price of stability in quality-aware federated learning.
\newblock In \emph{GLOBECOM 2023 - 2023 IEEE Global Communications Conference}, pages 734--739, 2023.
\newblock \doi{10.1109/GLOBECOM54140.2023.10437743}.

\end{thebibliography}
\appendix
\section{Relaxation of the optimal aggregation problem.}\label{appendix:relaxation}

Let $\collab\in\defEnsLigne{0,1}^J$ and the associated coalition $\cB = \defEnsLigne{j\in\agents:\:B_j = 1}$ be fixed. We present in this appendix a natural relaxation of problem \Cref{master_problem_in_text2} which motivates the choice $\sum_{j\in\cB}\nstar_j (\bfth)=\Nbar=(\noutside + 1)\abs{\cB}^{\frac{1}{1+\gamma}}-1$ in the simplified contribution scheme $(\Bstar, \bfns)$. The optimal aggregation problem with $\collab$ fixed is
\begin{equation*}
\quad\text{maximize}\quad \bfn\in\R^J\,\mapsto\,
\sum_{j\in\agents}u_j(\collab,\bfn)
\quad \text{subject to}\quad \begin{cases}
        \max_{j\in\cB}n_j - \maxcontrib(\collab, \bfn) \leq 0 \\
        \max_{j\in\agents}-n_j \leq 0 \eqsp,
    \end{cases}
\end{equation*}
Since agents outside of the coalition maximize their utility, the problem can equivalently be stated as
\begin{equation*}
(\cT_\collab):\quad\text{maximize}\quad \bfn\in\R^{\abs{\cB}}\,\mapsto\,\sum_{j\in\cB}u_j (\collab, \bfn) + \zeta_{\collab}
\quad \text{subject to}\quad \bfn \in\Xi_{\collab}\eqsp,
\end{equation*}
where
\begin{equation}\label{def:xib}
\Xi_\collab = \defEns{\bfn\in\R_{+}^{\cB}:\:\begin{array}{cc}
    \max_{j\in\cB}\defEns{n_{j} - \maxcontrib_j(\collab,\bfn)} \leq 0   \\
    \max_{j\in\cB}-n_{j} \leq 0
\end{array}}\quad\text{and}\quad \zeta_\collab = \sum_{j\notin \cB}\outside_j\eqsp.  
\end{equation}

We start by rewriting $(\cT_\collab)$ in a more convenient way. Instead of working with $\bfn\in\R_+ ^{\abs{\cB}}$ directly, we make appear (i) the total number of samples and (ii) the sharing out of samples between agents. Formally, for any $\bfn\in\R^{\abs{\cB}} _+$ consider $N=\bfn^{\transpose}\collab$ and $\blamb =N^{-1}\bfn\in\Delta_{\abs{\cB}}$, where $\Delta_{\abs{\cB}}$ is the simplex of dimension $\abs{\cB}$. Observe that the average type in the coalition reads  $\avgtype=\blamb^{\transpose}\bfth$. Moreover, the welfare evaluated in $(\collab, \bfn)$ rewrites
\begin{align*}
    \sum_{j\in\agents}u_j(\collab, \bfn) &=\sum_{j\in\agents}B_j \parentheseDeux{-a\parenthese{\Rstar+\eps(\collab, \blamb N)}-cn_{j}} + \sum_{j\in\agents}(1-B_j)\outside_j\\ &=-a\abs{\cB}\parenthese{\Rstar+\eps(\collab, \blamb N)} -cN +\zeta_{\collab}\eqsp,\\
    &= -a\abs{\cB}\parentheseDeux{\Rstar+2\parenthese{\alpha_{\delta}(1+N)^{-\gamma}+\beta + \blamb^{\transpose}\bfth}} -cN + \zeta_{\collab}\\
    &=\widetilde{W}_\collab (\blamb, N )\eqsp.
\end{align*}
We denote by $\widetilde{\Xi}_\collab =\defEns{(\blamb,N)\in\Delta_{\abs{\cB}}\times\R_{+}:\:\blamb N \in\Xi_\collab}$ and $\widetilde{\cT}_\collab$ the problem:
\begin{equation}\label{def:tildetb}
(\widetilde{\cT}_{\collab}):\quad\text{minimize}\quad (\blamb, N)\in\Delta_{\abs{\cB}}\times \R_+ \mapsto -\widetilde{W}_{\collab}(\blamb, N)\quad\text{subject to}\quad(\blamb, N)\in\widetilde{\Xi}_{\collab}\eqsp.    
\end{equation}
By definition, if $(\blamb, N)$ is a solution to $\widetilde{\cT}_{\collab}$, then $\bfn=\blamb N$ is a solution to $\cT_{\collab}$. We can further decompose $\widetilde{\cT}_\collab$ by observing that
\begin{equation*}
    -\widetilde{W}_{\collab}(\blamb, N)=f(N) + g(\blamb)\eqsp, 
\end{equation*}with $$f(N)= a\abs{\cB}(\Rstar + 2\alpha_{\delta}(1+N)^{-\gamma}+2\beta)+cN +\zeta_\collab \quad\text{and}\quad g(\blamb)= 2a\abs{\cB}\blamb^{\transpose}\bfth\eqsp.$$
Finally, we denote a slice of $\widetilde{\Xi}_{\collab}$ along $N\geq 0$ as $\widetilde{\Xi}_\collab (N) = \defEnsLigne{\blamb\in\Delta_{\abs{\cB}}:\: (\blamb, N)\in\widetilde{\Xi}_\collab}$ and $\cN=\defEnsLigne{N\geq 0:\: \widetilde{\Xi}_{\collab}(N)\ne \emptyset}$. $\widetilde{\cT}_\collab$ comes down to
\begin{equation}
    \label{def:decomposedtildet}
    \min_{(\blamb, N)\in\widetilde{\Xi}_{\collab}}\defEnsLigne{f(N)+g(\blamb)}=\min_{N\in\cN}\defEnsLigne{f(N)+\min_{\blamb\in\widetilde{\Xi}_{\collab}(N)}g(\blamb)}\eqsp.
\end{equation}A strategy to solve \Cref{def:decomposedtildet} is to (i) address the innermost problem $\min_{\blamb\in\widetilde{\Xi}_{\collab}(N)}g(\blamb)$ with $N\in\cN$ fixed, and denoting $\blamb^{(N)}\in\widetilde{\Xi}_{\collab}(N)$ its solution, (ii) solve the outermost problem: \begin{equation}
    \min_{N\in\cN}\defEnsLigne{f(N) + g(\blamb^{(N)})}\eqsp.
\end{equation}Point (i) is done in the proof of \Cref{prop:optimalfullinfo}. However, the resulting problem in (ii) is hard to tackle because $\blamb^{(N)}$ has no simple form. Therefore, we leave aside the term $g(\blamb^{(N)})$ (which can be easily controlled, as shown in \Cref{lemma:approxwelfare}) to determine $N$ and only consider the problem \begin{equation}\label{eq:relaxedpb}
    \min_{N\geq 0}\,f(N)\eqsp.
\end{equation}Since $f$ is differentiable and strictly convex, its minimizer is uniquely defined by $f'(\Nbar)=0$, which gives\begin{equation}\label{eq:defnbar}
    \Nbar = (\noutside + 1)\abs{\cB}^{\frac{1}{1+\gamma}}-1\eqsp,
\end{equation} where $\noutside$ is defined in \Cref{prop:outsideoption}. As the solution of the relaxed problem \eqref{eq:relaxedpb}, we take $\sum_{j\in\agents}\nstar_j (\bfth)=\Nbar$ is the simplified contribution scheme $(\Bstar, \bfns)$. In many reasonable settings, this approximation is very satisfactory as shown by \Cref{lemma:approxwelfare}.
\section{Proofs}\label{proofs}
\excessrisk*
\begin{proof}
   \begin{enumerate}[wide, labelwidth=!, labelindent=0pt,label=(\roman*)]
    \item 
Let $\defEnsLigne {(X^j _1, Y^i_1),\ldots, (X^j _{n_{j}},Y^j _{n_{j}})}$ be $n_{j}>0$ i.i.d samples from $P_j \in\cP$ and $\agentmodel{j}=\inf_{g\in\hypoclass}\emprisk{g}{j}$. First, observe that
        \begin{equation}
            \label{risk_Pzero_pi}
            \risk{\agentmodel{j}}{0}\leq \risk{\agentmodel{j}}{j}+\sup_{g\in\hypoclass}\abs{\risk{g}{0}-\risk{g}{j}}=\risk{\agentmodel{j}}{j} + \type_j\eqsp.
        \end{equation}

        Let $\upsilon >0$ and $g_j^{\upsilon} \in \hypoclass$ such that $\risk{g^\upsilon _j}{j} \leq \inf_{g\in\hypoclass} \risk{g}{j} + \upsilon$. Then with probability at least $1-\delta$,
        \begin{align}
            \risk{\agentmodel{j}}{j}&=\risk{\agentmodel{j}}{j}-\risk{g_j ^\upsilon}{j} + \risk{g_j ^\upsilon}{j} \notag \\
            &\leq \parenthese{\emprisk{\g^\upsilon _j}{j}-\emprisk{\agentmodel{j}}{j}} + \risk{\agentmodel{j}}{j}-\risk{g_j ^\upsilon}{j} + \risk{g_j ^\upsilon}{j} \notag \\
            &\leq 2\sup_{g\in\hypoclass}\abs{\risk{g}{j}-\emprisk{g}{j}} + \inf_{g\in\hypoclass}\risk{g}{j}+\upsilon \notag \\
            &\leq 2\parenthese{\frac{\alpha_{\delta}}{(n+1)^\gamma}+\beta} + \inf_{g\in\hypoclass}\risk{g}{j} \label{bound_risk_ghat}\eqsp,
        \end{align}
        where the last inequality holds taking the limit $\upsilon\rightarrow 0$ and using \cref{assumption_upper_bound_risk}. Finally, observing that $\inf_{g\in\hypoclass}\risk{g}{j} \,-\, \Rstar=\inf_{g\in\hypoclass}\risk{g}{j}-\inf_{g\in\hypoclass}\risk{g}{0}\leq  \sup_{g\in\hypoclass}\abs{\risk{g}{j}-\risk{g}{0}}= \type_j$, and combining \eqref{risk_Pzero_pi} as well as \eqref{bound_risk_ghat} yields
        \begin{equation}\label{bound_risk_ghat_Pzero}
            \risk{\agentmodel{j}}{0}\leq \Rstar + 2\type_j + 2\parenthese{\frac{\alpha_{\delta}}{(n+1)^\gamma}+\beta}=\Rstar + \eps(\type_j, n)\eqsp,
        \end{equation}
        with probability at least $1-\delta$. 
      \item
        Let $\collab\in\defEns{0,1}^J$ and $\bfn= (n_1,\ldots, n_{j})\in\R_+ ^J$. For any $j \in\collab$, let $\defEnsLigne {(X^j _1, Y^i_1),\ldots, (X^j _{n_{j}},Y^j _{n_{j}})}$ be a collection of $n_{j}>0$ i.i.d samples from $P_j\in\cP$. Denote by $N=\sum_{j\in \agents} B_j n_{j}$ and  consider $\blamb = (\lambda_1,\ldots,\lambda_J)$ such that $\lambda_j = B_j n_{j} / N$. Note that $\blamb$ belongs to  the $J$-dimensional simplex. For any $g\in\hypoclass$, the empirical risk over contributions is 
        \begin{equation}
            \label{coalition_emprisk}
            \emprisk{g}{\collab} =\frac{1}{N}\sum_{j\in\agents}B_j \sum_{i=1}^{n_{j}} \ell(g(X^j_i), Y^j_i)=\sum_{j\in\agents}B_j \lambda_j \emprisk{g}{j}\eqsp ,
        \end{equation}and the population risk is
        \begin{equation}
            \label{exp_coalition_emprisk}
            \cR_{\collab} (g) = \E\parentheseDeux{\emprisk{g}{\collab}}=\sum_{j\in\agents}B_j\lambda_j \E\parentheseDeux{{\emprisk{g}{j}}}=\sum_{j\in\agents}B_j\lambda_j \risk{g}{j}\eqsp.
        \end{equation}
        One the one hand, the collaborative model $\collabmodel\in\hypoclass$ satisfies:
    \begin{align}
    \nonumber \risk{\collabmodel}{0}\leq \cR_{\collab} (\collabmodel) + \sup_{g\in\hypoclass}\abs{\risk{g}{0}-\cR_{\collab} (g)}\notag
    &\leq \cR_{\collab} (\collabmodel) + \sup_{g\in\hypoclass}\defEns{\sum_{j\in\agents}\lambda_j \abs{\risk{g}{j}-\risk{g}{0}}} \notag \\
    \label{risk_vs_Pzero_mixture} & \leq \risk{\collabmodel}{\collab} + \sum_{j\in\agents}\lambda_j \sup_{g\in\hypoclass}\abs{\risk{g}{j}-\risk{g}{0}} \notag \\
    & = \cR_{\collab} (\collabmodel) + \avgtype \eqsp,
\end{align}

with $\avgtype= N^{-1}\sum_{j\in\agents}B_j n_{j} \type_j = \sum_{j\in\agents}B_j\lambda_j \type_j$. Now, let $\upsilon > 0$ and $g^\upsilon\in\hypoclass$ such that $\cR_{\collab} (g^\upsilon) \leq \inf_{g\in\hypoclass}\cR_{\collab} (g) +\upsilon$. We also have
\begin{align}
    \nonumber \cR_{\collab} (\collabmodel) &= \cR_{\collab} (\collabmodel) -\cR_{\collab} (g^\upsilon) + \cR_{\collab} (g^\upsilon) \notag\\
    \nonumber & \leq \parenthese{\emprisk{g^\upsilon}{\collab}-\emprisk{\collabmodel}{\collab}} + \cR_{\collab} (\collabmodel) -\cR_{\collab} (g^\upsilon) + \cR_{\collab} (g^\upsilon) \notag \\
    \nonumber & \leq 2\sup_{g\in\hypoclass}\abs{\cR_{\collab} (g) - \emprisk{g}{\collab}} + \inf_{g\in\hypoclass}\cR_{\collab} (g) + \upsilon \notag \\
    \label{pac_bound_weighted_risk_final_mixture} & \leq 2\parenthese{\frac{\alpha_{\delta}}{(N+1)^{\gamma}}+\beta}+\inf_{g\in\hypoclass} \cR_{\collab} (g) \quad \text{with probability at least }1-\delta \eqsp,
\end{align}
where the first inequality comes from $\emprisk{\collabmodel}{\collab}=\inf_{g\in\hypoclass}\emprisk{g}{\collab}\leq \emprisk{g^\upsilon}{\collab}$, and the last is obtained by taking $\upsilon\rightarrow 0$ and applying assumption \cref{assumption_upper_bound_risk}. Now, observe that $\inf_{g\in\hypoclass}\cR_{\collab} (g) -\Rstar = \inf_{g\in\hypoclass}\cR_{\collab} (g) -\inf_{g\in\hypoclass}\risk{g}{0}\leq\sup_{g\in\hypoclass}\abs{\cR_{\collab} (g) - \risk{g}{0}}\leq\sum_{j}B_j\lambda_j\sup_{g\in\hypoclass}\abs{\risk{g}{j}-\risk{g}{0}}=\avgtype$. Combining this observation with \cref{risk_vs_Pzero_mixture} and \cref{pac_bound_weighted_risk_final_mixture} gives with probability $1-\delta$
   \begin{equation}
       \label{risk_h_Pzero}
       \risk{\collabmodel}{0}\leq \Rstar + 2\parenthese{\frac{\alpha_{\delta}}{(N+1)^{\gamma}}+\beta} + 2\avgtype = \Rstar + \eps(\vartheta(\collab, \bfn), N)
     \eqsp.\end{equation}
   \end{enumerate}
 \end{proof}

\outsideoption*

\begin{proof}
For any $j\in\agents$ and $n\ge 0$, \Cref{lemma:excessrisk} gives that $$-u_j \parenthese{(0,\collab_{-j}), (n,\bfn_{-j}),\type_j} =a[\Rstar + 2(\alpha_{\delta}(n+1) ^{-\gamma}  + \beta + \type_j)] + cn= f(n)\eqsp.$$Since $f$ is strictly convex and differentiable, it admits a unique maximizer $\noutside\geq0$ determined by $f'(\noutside) =0$. Simple algebra leads to $\noutside = (2a\gamma c^{-1}\alpha_{\delta})^{1/\gamma+1}-1$.
\end{proof}

\optimalfullinfo*

\begin{proof}
Recall that the optimal aggregation problem reads
\begin{align}
\label{master_problem} 
    &\text{maximize}\quad W(\collab,\bfn)\in\defEns{0,1}^J \times \R^J _+ \mapsto \sum_{j\in\agents}u_j (\collab, \bfn)\\
    &\text{subject to}\quad
        \min_{j\in\agents}u_j ( \collab, \bfn)- \outside_j \geq 0\eqsp.
    \notag
\end{align}
Define for any $j\in\agents, \collab_{-j}\in\defEns{0,1}^{J-1}$, and $\bfn_{-j}\in\R_{+}^{J-1}$, the maximum number of samples agent $j\in\agents$ having $B_j=1$ may be asked given their participation constraint: $$\maxcontrib
_j(\collab, \bfn)=\max\defEnsLigne{n\geq0:\,u_j \parenthese{(1,\collab_{-j}), (n,\bfn_{-j})} \geq \outside _j}\eqsp.$$Given \Cref{definition_utility} and \Cref{prop:outsideoption}, we obtain 
\begin{equation}
    \label{definition:maxcontrib}
    \maxcontrib_j(\collab,\bfn)=\noutside-ac^{-1}(\eps(\collab, \bfn)-\eps(\type_j, \noutside))\eqsp,
\end{equation}
where $N=\sum_{j\in\agents}B_j n_j = \bfn^{\transpose} \collab$. Problem \eqref{master_problem} rewrites in canonical form
\begin{equation}\label{master_problem_canonical}
\text{minimize}\quad -W(\collab, \bfn)
\quad \text{subject to}\quad \begin{cases}
        \max_{j\in\cB}n_j - \maxcontrib_j (\collab, \bfn) \leq 0 \\
        \max_{j\in\agents}-n_j \leq 0 \eqsp,
    \end{cases}
\end{equation}
where $\cB=\defEns{j\in\agents:\:B_j=1}$. We first show that it admits a unique minimum.

Fix a $\collab\in\defEns{0,1}^{J}$ and consider problem \eqref{master_problem_canonical} with respect to $\bfn\in\R^J _+$ only. We call this subproblem $\cT_{\collab}$. First, we show that it is enough to focus on $(n_j)_{j\in\cB}$ rather than $(n_j)_{j\in\agents}$ in $\cT_\collab$. With $\bfn_{\cB} = (\n_{j})_{j\in\cB}$ and $\bfn_{\cB^c}=(\n_{j})_{j\notin\cB}$, note that $W(\collab,\bfn)$ is decomposable in $\bfn_\cB$ and $\bfn_{\cB^c}$:
$$
W(\collab, \bfn) = W(\collab, \bfn_\cB) + W(\collab, \bfn_{\cB^c})\eqsp,
$$
where, by a slight abuse of notation $W(\collab, \bfn_{\cB})=\sum_{j\in\collab} u_j(\collab,\bfn)$ and $W(\collab, \bfn_{\cB^c})=\sum_{j\in\cB^c}u_j(\collab,\bfn)$. With
\begin{equation}\label{def:xib}
\Xi_\collab = \defEns{\bfn_{\cB}\in\R_{+}^{\cB}:\:\begin{array}{cc}
    \max_{j\in\cB}\defEns{n_{j} - \maxcontrib_j(\collab,\bfn)} \leq 0   \\
    \max_{j\in\cB}-n_{j} \leq 0
\end{array}}\eqsp,    
\end{equation}

problem $\cT_\collab$ is equivalent to
\begin{equation}
    \label{problem:decomposable}
    \min_{\bfn_{\cB}\in\Xi_\collab}-W(\collab, \bfn_{\cB}) + \min_{\bfn_{\cB^c}\in\R_{+}^{J-\abs{\cB}}}-W(\collab, \bfn_{\cB^c}) = \min_{\bfn_{\cB}\in\Xi_\collab}-W(\collab, \bfn_{\cB}) + \sum_{j\in\agents}-(1-B_j)\outside_j\eqsp,
\end{equation}
by \Cref{prop:outsideoption}. Since $\sum_{j\in\agents}-(1-B_j)\outside_j$ is constant, by \Cref{problem:decomposable} it is enough to focus on the existence of $\min_{\bfn_{\cB}\in\Xi_\collab}-W(\collab, \bfn_{\cB})$. On the one hand, $\Xi_\collab$ is bounded. Indeed for any $j\in\cB$, by \Cref{definition:maxcontrib} and \Cref{lemma:excessrisk}

\begin{align}
\maxcontrib_j(\collab, \bfn_{\cB})&=\noutside -\frac{a}{c}\parentheseDeux{\eps(\collab,\bfn_{\cB})-\eps(\type_j,\noutside)} \notag\\
        &= \noutside-\frac{a}{c}\parentheseDeux{\Rstar + 2\parentheseDeux{\frac{\alpha_{\delta}}{(1+\bfn_{\cB}^{\transpose}\collab)^\gamma}+\beta + \avgtype} - \Rstar -2\parentheseDeux{\frac{\alpha_{\delta}}{(1+\noutside)^{\gamma}}-\beta-\type_j}}\notag\\
        &=\noutside -\frac{2a}{c}\parentheseDeux{\alpha_{\delta}\parenthese{(1+\bfn_{\cB}^{\transpose}\collab)^{-\gamma}-(1+\noutside)^{-\gamma}}+\parenthese{\avgtype-\type_j}} \notag\\
        &\leq \noutside - \frac{2a}{c}\parentheseDeux{\alpha_{\delta}\parenthese{1-(1+\noutside)^{-\gamma}}+\max_{j\in\cB}\type_j}\notag\\
        &= M \label{eq:gjbounded}
\end{align}
Thus, $\Xi_\collab \subset [0,M]^{\abs{\cB}}$. Moreover, $\Xi_\collab$ is closed and convex. For any $j\in\cB$, rewrite
    \begin{align*}
        \maxcontrib_j(\collab,\bfn_{\cB})
        &=\noutside -\frac{2a}{c}\parentheseDeux{\alpha_{\delta}\parenthese{(1+\bfn_{\cB}^{\transpose}\collab)^{-\gamma}-(1+\noutside)^{-\gamma}}+\parenthese{\parenthese{\bfn_{\cB}^{\transpose}\collab}^{-1}\bfn_{\cB}^{\transpose}\bfth-\type_j}} \\
        &= g_j (\bfn_{\cB})\eqsp,
    \end{align*}

    and define $h_j:\bfn_{\cB}\mapsto \n_{\cB, j}-g_j(\bfn_{\cB})$. Observe that $\Xi_\collab =\defEnsLigne{\bfn_{\cB}\in\R_{+}^{\abs{\cB}}:\: \forall j \in\cB,\,  n_{\cB, j}\geq 0\:\text{ and }\:h_j (\bfn_{\cB})\leq 0}$, that is $\Xi_\collab = A\,\cap\, B$ with $A=\R_{+}^{\abs{\cB}}$ and $B=h^{-1}_1 ((-\infty,0])\cap\ldots\cap h^{-1}_{\abs{\cB}} ((-\infty,0])$. For any $j\in\cB$, $h^{-1}_j ((-\infty,0])$ is convex, because so is $h_j$. Additionally, $h^{-1}_j ((-\infty,0])$ is closed as the inverse image of a closed set by a continuous function. It follows that $\Xi_\collab$ is convex and closed as the intersection of convex and closed sets. 

As a consequence $\Xi_\collab$ is compact and convex, and $-W(\collab, \centraldot)$ is strictly convex so $\cT_{\collab}$ admits a unique minimizer $\bfn_{\cB}^\opt \in\Xi_\collab$.

Since there are finitely many $\collab\in\defEns{0,1}^{J}$, $\min\defEnsLigne{-W(\collab, \bfn_{\cB}^\opt)+\sum_{j\in\cB^c}\outside_j,\, \collab\in\defEns{0,1}^{J}}>0 $ exists and coincide with the minimum of problem $\eqref{master_problem_canonical}$. We show later in the proof that the optimal $\collab^\opt \in\defEns{0,1}^J$ is unique (see \textit{part 1}), so the minimizer of \Cref{master_problem_canonical} is unique. This establishes the point (i) of the result.

The remainder of the proof proceeds as follows: we first characterize the optimal $\collab^\opt\in\defEns{0,1}^J$, and then prove that the solution $\bfn^\opt _{\cB^\opt} \in \R_{+}^{\abs{\cB}}$ of $\cT_{\collab^\opt}$ has the form presented in point (ii) of the result.

\proofpart{1}{Optimal $\collab\in\defEns{0,1}^{\agents}$}We show that $\collab^\opt=(1,\ldots,1)^{\transpose}$. By contradiction, assume there exists $\collab' \ne (1,\ldots,1)^{\transpose}$ and $\bfn' \in\R^J _+$ such that for any $(\collab, \bfn)\in\defEns{0,1}^J \times \R^J _+$ 
\begin{equation}
    \label{eq:bprimeoptimal}
    -W(\collab', \bfn')\leq -W(\collab, \bfn )\eqsp.
\end{equation}
Denote by $j\in\agents$ an agent such that $B'_j = 0$, and first assume $\type_j \geq \vartheta(\collab', \bfn')$. Consider the alternative allocation $(\collab'' , \bfn'')=((1,\collab'_{-j}), (0, \bfn'_{-j}))$. Observe that $B'_j n'_j = B''_j n''_j = 0$, so
\begin{equation}
N'=\collab'^{\transpose}\bfn'=\collab''^{\transpose}\bfn''=N''\quad\text{and}\quad\vartheta(\collab', \bfn')  
=\vartheta(\collab'', \bfn'')\eqsp.\label{eq:avgtypeisthesame}
\end{equation}
This in particular implies 
\begin{align}
    \eps(\collab', \bfn')&=2\parentheseDeux{\alpha_\delta (1+N')^{-\gamma}+\beta+\vartheta(\collab', N')}\notag\\
    &=2\parentheseDeux{\alpha_\delta (1+N'')^{-\gamma}+\beta+\vartheta(\collab'', N'')}=\eps(\collab'', \bfn'')\eqsp.\label{eq:excessriskarethesame}
\end{align}
Thus by \eqref{eq:excessriskarethesame}
\begin{align}
    -u_j(\collab', \bfn') = \outside_j &= a(\Rstar + \eps(\type_j, \noutside)) + c\noutside \notag >a(\Rstar + \eps(\type_j, \noutside))\notag\\
    &\geq a(\Rstar + \eps(\collab', \bfn')) = a(\Rstar + \eps(\collab'', \bfn''))= -u_j(\collab'',\bfn'')\label{eq:jisbetterofff}\eqsp,
\end{align}where the second inequality results from $\type_j  \geq \vartheta(\collab', \bfn')$ and $N'\geq\noutside$. Finally by definition of $(\collab'', \bfn'')$, \eqref{eq:avgtypeisthesame} and  \eqref{eq:excessriskarethesame}, for any $k\ne j$ such that $B''_k = 1$:
\begin{align}
    -u_k (\collab'', \bfn'')&=a\parentheseDeux{\Rstar + 2\parenthese{\alpha_{\delta}(1+N'')^{-\gamma}+\beta+\vartheta(\collab'', \bfn'')}}+cn_k '' \notag\\
    &=a\parentheseDeux{\Rstar + 2\parenthese{\alpha_{\delta}(1+N')^{-\gamma}+\beta+\vartheta(\collab', \bfn')}}+cn_k ' = -u_k(\collab', \bfn')\eqsp.\label{eq:utilitysameforotherss}
\end{align}
Hence, \eqref{eq:utilitysameforotherss} together with  \eqref{eq:jisbetterofff} give
\begin{align}
    -W(\collab'', \bfn'' )=-\sum_{k\ne j}u_{k}(\collab'',\bfn'')  - u_j(\collab'',\bfn'')&= -\sum_{k\ne j}u_i(\collab', \bfn') -u_j(\collab'', \bfn'')\notag \\
    &<-\sum_{k\ne j}u_i(\collab', \bfn') -u_j(\collab', \bfn')=-W(\collab', \bfn')\eqsp.
\end{align}
This contradicts \eqref{eq:bprimeoptimal}. Now assume $\type_j < \vartheta(\collab', \bfn')$, and let $R\in\agents$ be such that $\type_R = \max\defEnsLigne{\type_k,\, k\in\cB}$. Consider $(\collab'', \bfn'')\in\defEns{0,1}^J \times \R^J _+$ where
\begin{align*}
    \collab'' &= (B'_1, \ldots, B'_{j-1},\,1\,,B'_{j+1},\ldots,B'_J)^{\transpose} \\
    \text{and } \bfn''&=(n'_1 , \ldots, n'_{j-1},\, \maxcontrib_j\,,n'_{j+1},\ldots, n'_{R-1}, \,\max(0,n'_{R} - \maxcontrib_j)\,,n'_{R+1},\ldots,n'_J)^{\transpose}\eqsp.
\end{align*}
which is feasible. Observe on the one hand that
\begin{equation*}
    N'' = N'\quad\text{ and}\quad\vartheta(\collab'', \bfn'') \leq\vartheta(\collab', \bfn') + \frac{\maxcontrib_j}{N'}\parenthese{\type_j - \type_{R}} < \vartheta(\collab', \bfn')\eqsp,
\end{equation*}
because $\type_j < \vartheta(\collab', \bfn')\leq\type_R$. In particular
\begin{align}    
    \riskexcess(\collab'', \bfn'')&=\Rstar +2\parentheseDeux{\alpha_\delta (N'' + 1)^{-\gamma}+\beta + \vartheta(\collab'', \bfn'')} \notag\\
    &<
    \Rstar +2\parentheseDeux{\alpha_\delta (N' + 1)^{-\gamma}+\beta + \vartheta(\collab', \bfn')}=\riskexcess(\collab', \bfn')\eqsp.\label{eq:notjbetteroff}
\end{align}\eqref{eq:notjbetteroff} in turn implies that for any $k\ne j$ such that $B'' _k =1$:
\begin{equation}
    \label{eq:notjbetterofftwice}
    -u_k (\collab'', \bfn'')=a\parentheseDeux{\Rstar + \eps(\collab'', \bfn'')}+c n''_k < a\parentheseDeux{\Rstar + \eps(\collab', \bfn')}+c n'_k =-u' _k (\collab', \bfn')\eqsp.
\end{equation}
because $n'_k = n''_k$ for $k\in\cB\setminus\defEnsLigne{j,R}$ and $n'' _R < n' _R$. On the other hand, we have
\begin{equation}
    \label{eq:ujindifferent}
    u_j(\collab'', \bfn'')=\outside_j=u_j(\collab',\bfn')\eqsp,
\end{equation}
by definition \eqref{definition:maxcontrib} of $\maxcontrib_j (\collab'', \bfn'')$. \eqref{eq:notjbetterofftwice} and \eqref{eq:ujindifferent} together yield
\begin{align*}
    -W(\collab'',\bfn'' ) = -\sum_{k\ne j}u_k(\collab'', \bfn'') - u_j(\collab'', \bfn'')
    &=-\sum_{k\ne j}u_k(\collab'', \bfn'') - u_j(\collab', \bfn') \\ 
    &< -\sum_{k\ne j}u_k(\collab', \bfn') - u_j(\collab', \bfn')
    =-W(\collab', \bfn' )\eqsp,
\end{align*}
which once again violates \eqref{eq:bprimeoptimal}. Thus, we obtain $\collab^\opt=(1,\ldots,1)^{\transpose}$, that is $\cB^\opt=\agents$.

\proofpart{2}{Characterization of $\bfnopt_{\agents}\in\R_{+}^{J}$}We now focus on the problem $\cT_{\collab^{\opt}}$. To lighten notations, we write $\bfn^\opt \in\R_+ ^J$ instead of $\bfn^\opt _{\cB^{\opt}}$, and $\collab$ instead of $\collab^\opt = (1,\ldots,1)$. We first reformulate $\cT_\collab$ in a more convenient way. Recall that $\cT_{\cB}$ reads
\begin{equation}\label{def:pbct}
    (\cT_\collab):\quad\text{minimize}\quad -W(\collab, \bfn)\quad\text{subject to}\quad \bfn \in\Xi_{\collab}\eqsp,
\end{equation}where $\Xi_{\collab}$ is defined in \eqref{def:xib}. As explained in \Cref{appendix:relaxation}, this problem can equivalently be stated as 
\begin{equation}\label{def:tildetb}
(\widetilde{\cT}_{\collab}):\quad\text{minimize}\quad -\widetilde{W}_{\collab}(\blamb, N)\quad\text{subject to}\quad(\blamb, N)\in\widetilde{\Xi}_{\collab}\eqsp,    
\end{equation}
where
\begin{equation*}
    -\widetilde{W}_{\collab}(\blamb, N)=f(N) + g(\blamb)\eqsp, 
\end{equation*}with $$f(N)= aJ(\Rstar + 2\alpha_{\delta}(1+N)^{-\gamma})+cN\quad\text{and}\quad g(\blamb)= 2aJ\blamb^{\transpose}\bfth\eqsp,$$and$$\widetilde{\Xi_\collab} = \defEns{(\blamb, N)\in\Delta_{\abs{\cB}}\times\R_+ :\: \blamb N \in \Xi_{\collab}}\eqsp. $$
Writing $\widetilde{\Xi}_\collab (N) = \defEnsLigne{\blamb\in\Delta_{J}:\: (\blamb, N)\in\widetilde{\Xi}_\collab}$ for $N\geq 0$ and  $\cN=\defEnsLigne{N\geq0 :\:\widetilde{\Xi}_{\collab}(N)\ne\emptyset}$,  $\widetilde{\cT}_\collab$ comes down to
\begin{equation*}
    \min_{(\blamb, N)\in\widetilde{\Xi}_{\collab}}\defEnsLigne{f(N)+g(\blamb)}=\min_{N\in\cN}\defEnsLigne{f(N)+\min_{\blamb\in\widetilde{\Xi}_{\collab}(N)}g(\blamb)}\eqsp.
\end{equation*}
In this proof, we address the innermost problem, which is enough to show that $\bfn^\opt$ satisfies the point (ii) of the result. Let $N\geq 0$ such that $\widetilde{\Xi}_{\collab}(N)\ne \emptyset$, which exists because $\Xi_{\collab}\ne \emptyset$ (for instance, $(\parentheseLigne{1,0,\ldots,0}, (\noutside,0,\ldots,0))\in\Xi_\collab $). By definition of $\widetilde{\Xi}_{\collab}$, it reads\begin{equation}\label{workable_problem}(\widetilde{\cT}^{(N)}):\quad\text{minimize}\quad\blamb\in\Delta_{J}\mapsto 2aJ\blamb^{\transpose}\bfth \quad\text{subject to}\quad\begin{cases}
        \forall j \in\cB:\lambda_j \geq 0 \\
        \forall j\in\cB: \: \lambda_j  N\leq \maxcontrib_j (\collab,\blamb N) \\
        \sum_{j\in\cB} \lambda_j = 1\eqsp,
    \end{cases}
\end{equation}

The following lemma provides a necessary condition for any solution to problem \eqref{workable_problem}.
\begin{lemma}
\label{lemma_charact_lambda}
Let $\blamb^{(N)}=(\lambda^{(N)}_1,\ldots,\lambda^{(N)}_J)\in\Delta_{J}$ be a solution to \eqref{workable_problem}. If there exists $r\in\agents$ such that $\lambda^{(N)}_{r} > 0$, then $\lambda_{k} ^{(N)} =  N^{-1}\maxcontrib_{k} (\collab, N\blamb^{(N)})$ for any $k<r$.
\end{lemma}
\begin{proof}We denote by $\bar{\bfn}(\collab, \bfn)=(\maxcontrib_{1}(\collab,\bfn),\ldots,\maxcontrib_{J}(\collab, \bfn))^{\transpose}$ for any $\bfn\in\R_+ ^{J}$.
The Lagrangian associated to problem \eqref{workable_problem}, with $\boldsymbol{\mu}, \boldsymbol{\rho}$ and $\nu$ the associated dual variables, reads:
\begin{align*}
    \cL (\blamb, \bfmu, \bfrho, \nu) &= 2aJ\blamb^{\transpose}\bfth -\bfmu^{\transpose}\blamb +\nu\parenthese{1-\blamb^{\transpose}\bOne} +\bfrho^{\transpose}\parentheseDeux{ N\blamb - \Bar{\bfn}(\collab, N\blamb)} \\
    &=2aJ\blamb^{\transpose}\bfth +c N -\bfmu^{\transpose}\blamb +\nu\parenthese{1-\blamb^{\transpose}\bOne}\\
    &+\bfrho^{\transpose}( N\blamb - \noutside\bOne +\frac{2a}{c}\parentheseDeux{\alpha_{\delta}\parenthese{ (N +1)^{-\gamma}-(\noutside+1)^{-\gamma}}\bOne+\parenthese{(\blamb^{\transpose}\bfth) \bOne-\bfth}}) \eqsp.
\end{align*}

Since the objective and the constraints are convex, $\cL$ admits a saddle point $(\blamb^{(N)}, \bfmu, \bfrho, \nu)\in\Delta\times\R^{J} \times \R^{J} \times \R$ which is solution to problem \eqref{workable_problem} and verifies the following KKT conditions:
\begin{numcases}{}
    \forall k \leq J:\: 2aJ\type_{k} - \mu_{k} + \rho_{k} ( N +2ac^{-1}\type_{k}) = \nu \label{1}  \\
    \forall k \leq J:\:\mu_{k} \lambda^{(N)}_{k} = 0,\quad\mu_{k}\geq 0 \label{2}\\
    \forall k\leq J:\:\rho_{k}\parenthese{\lambda^{(N)}_{k} \, N-\maxcontrib_{_k}(\collab, \blamb^{(N)}N)}=0,\quad \rho_{k} \geq 0 \label{4} \\
    \sum_{k\in\cB}\lambda^{(N)}_{k} = 1\eqsp.
\end{numcases}
 Assume there exists $r\in\cB$ such that $\lambda^{(N)}_{r} > 0$. By complementary slackness \eqref{2}, $\mu_{r} = 0$, and \eqref{1} gives for any $k<r$:
\begin{equation}
    \label{optimality_gamma}
    2aJ\type_{k} - \mu_{k} + \rho_{k} ( N +2ac^{-1}\type_{k}) = 2aJ\type_{r} + \rho_{r} ( N +2ac^{-1}\type_{r})\eqsp,
\end{equation}

That is 
    $$
   \rho_{k} = \frac{2aJ(\type_{r} - \type_{k})+\mu_{k}}{ N +2ac^{-1}\type_{k}} + \rho_{r}\frac{N +2ac^{-1}\type_{r}}{N +2ac^{-1}\type_{k}}\eqsp.
   $$
   By assumption $\type_{r} > \type_{k}$, and since $\mu_{k}\geq 0$ as well as $\rho_{r}\geq 0$, we have $\rho_{k}>0$. It follows from \eqref{4} for agent $k$ that $\lambda^\opt _{k}  N = \maxcontrib(\collab,N\blamb^{(N)} )$. 

\end{proof}Now, define
\begin{equation}
    \label{definition_w}M_{j} = \sum_{k\leq j}\maxcontrib_{k}(\collab, \blamb^{(N)}N)\quad\text{and}\quad
    L=\min\set{j\in\cB}{M_{j} \geq  N} \eqsp.
\end{equation} We show that 
\begin{equation}
    \label{eq:valuenotk}
\lambda^{(N)}_{k} = \maxcontrib_{k}(\collab , N\blamb^{(N)}) N^{-1}\text{ for any }k < L\quad\text{ and }\quad\lambda^{(N)}_{k} = 0\text{ for any }k > L\eqsp.
\end{equation}
To prove the first point, assume by contradiction that there exists $\ell < L$ such that $\lambda^{(N)}_{\ell} < \maxcontrib_{\ell}(\collab , \blamb^{(N)}N)N^{-1}$. By the contrapositive of \lemref{lemma_charact_lambda}, $\lambda^\opt_{\ell + 1}=\ldots=\lambda^\opt _{J}= 0$. Thus

$$
\sum_{k\leq J} \lambda^{(N)}_{k} = \sum_{k\leq\ell}\lambda^{(N)} _{k} \leq \sum_{k\leq \ell} \maxcontrib_{k}(\collab, \blamb^{(N)} N) N^{-1} =M _{\ell}  N^{-1}< 1\eqsp,
$$
by definition of $L$ \Cref{definition_w} and $\ell < L$. This violates the third constraint of problem \eqref{workable_problem}. For the second point, assume there exists $\ell> L$ such that $\lambda^{(N)}_{\ell} > 0$. By \Cref{lemma_charact_lambda}, $\lambda^{(N)}_{k} = \maxcontrib_{k}(\collab, \blamb^{(N)}N) N^{-1}$ for any $k\leq L < \ell$, so
\begin{align*}
    \sum_{k\leq J}\lambda^{(N)}_{k} \geq \sum_{k\leq L}\lambda^{(N)}_{k} + \lambda^{(N)}_{\ell} &= \sum_{k\leq L}\maxcontrib_k(\collab, \blamb^{(N)}N) N^{-1} +\lambda^{(N)}_{\ell} \\
    &= M _{L} N^{-1}+ \lambda^{(N)}_{\ell}\\
    &\geq 1 + \lambda^{(N)}_{\ell} > 1 \eqsp,
\end{align*}
which once again violates the constraints of the problem. Finally by \eqref{eq:valuenotk},
\begin{equation}\label{eq:valuek}
    \lambda^{(N)}_{i_L} = 1 -\sum_{k<L}\lambda^{(N)}_{k} -\sum_{k>L}\lambda^{(N)}_{k} = 1 - M _{L-1} N^{-1}\eqsp,
\end{equation}
by \eqref{eq:valuenotk}. By \eqref{eq:valuenotk} and \eqref{eq:valuek} we obtain for any $j\in\agents$:
\begin{equation}\label{eq:noptb}
 \lambda_j^{(N)} = N^{-1}\maxcontrib_{j}(\collab,N\blamb^{(N)})\2{j<L} + ( 1-N^{-1}M_{j-1})\2{j=L}\eqsp,
\end{equation}
Now, consider the solution $\bfn^\opt (\bfth)\in\R^J$ to the initial problem \eqref{workable_problem}, and define $$\begin{cases}
     N^\opt &= \sum_{j\in\agents}n_j^\opt (\bfth)\\
     \blamb^\opt (\bfth)&=\frac{1}{N^{\opt}}\bfn^\opt (\bfth)  \\
     L^\opt &= \min\defEnsLigne{j\in\agents:\: \sum_{k\leq j}\maxcontrib_k (\collab, \bfn^\opt (\bfth))\geq N^\opt}\eqsp.
\end{cases}$$

$\blamb^\opt (\bfth)$ satisfies \eqref{eq:noptb}, and multiplying by $N^\opt$ yields for any $j\in\agents$:
\begin{equation}
    n_j ^\opt (\bfth) = \maxcontrib _j (\collab, \bfn^\opt (\bfth))\2{j< L^\opt} + (N^\opt - \sum_{k\leq j}n_j ^\opt (\bfth))\2{j=L^\opt}\eqsp,
\end{equation}which establishes point (ii), and concludes the proof.

\end{proof}

\simplifyingfullinfo*
\begin{proof}
    By definition of the simplified scheme \eqref{def:simplifiedscheme}, $\nstar_j (\bfth) = \2{j\leq \Lstar}\maxcontrib_{j}(\bOne, \bfns(\bfth))=\2{j\leq \Lstar}\parentheseDeuxLigne{\noutside -  (2a/c)(\eps(\bOne, \bfns(\bfth)) - \eps(\type_j, \noutside))}$ for any $j\leq \Lstar$. Since $\sum_{j\in\agents}\nstar_j (\bfth) = \Nbar$, we get
    \begin{align*}
        \label{eq:summaxcontrib}
 \Nbar &= \Lstar\noutside - \frac{a}{c}\parentheseDeux{\Lstar\eps(\bOne, \bfns(\bfth)) - \sum_{j\leq \Lstar}\eps(\type_j, \noutside)} \notag \\
        &= \Lstar\parentheseDeux{\noutside - \frac{2a}{c}\parentheseDeux{\alpha_{\delta}\parenthese{(1+\Nbar)^{-\gamma}-(1+\noutside)^{-\gamma}}+(\vartheta(\bOne, \bfns(\bfth))-\empmean{\Lstar})}} \notag \eqsp,
    \end{align*}
    By \Cref{lemma:excessrisk} and denoting $\empmean{\Lstar}= L^{\star-1}\sum_{j\leq \Lstar}\type_j$. This yields $$\vartheta(\bOne, \bfns(\bfth)) = \empmean{\Lstar} - \alpha_{\delta}\parentheseLigne{(1+\Nbar)^{-\gamma}-(1+\noutside)^{-\gamma}}+(c/2a)(\noutside - \Nbar/\Lstar)\eqsp.$$
    Plugging back this value in $\maxcontrib^{\star}_j (\bOne, \bfns(\bfth))=\noutside -(2a/c)\parentheseDeuxLigne{\alpha_{\delta}\parenthese{(1+\Nbar)^{-\gamma}-(1+\noutside)^{-\gamma}}+\parenthese{\vartheta(\bOne, \bfns(\bfth))-\type_j}}$ gives for any $j\in\agents$:
    \begin{equation}\label{eq:nstarpeculiar}
            \nstar_j (\bfth) = \2{j\leq \Lstar}\parentheseDeux{\frac{\Nbar}{\Lstar}-\frac{2a}{c}(\empmean{\Lstar}-\type_j)}\eqsp.
    \end{equation}
    We now determine the order of magnitude of $\Lstar\in\iint{1}{J}$. Recall on the one hand that by \eqref{eq:gjbounded}, there exists $M>0$ such that $0 \leq \nstar_j (\bfth) \leq M$ for any $j\in\agents$, and on the other hand that $\empmean{\Lstar}-\type_j \leq \typemax-\typemin=\text{diam}(\types)<\infty$. Therefore by \eqref{eq:nstarpeculiar},
    \begin{equation}\label{eq:NstarLbounded}
            \frac{2a}{c}\text{diam}(\types) \leq \frac{\Nbar}{\Lstar} \leq M + \frac{2a}{c}\text{diam}(\types) \quad\text{so}\quad \frac{\Nbar}{\Lstar}=\Theta (1)\eqsp.
    \end{equation}

    Since $\Nbar = (\noutside+1)\abs{\cB^\star}^{\frac{1}{1+\gamma}}-1$ and $\cB^\star=\agents$,  \eqref{eq:NstarLbounded} results in $L=\Theta(J^{\frac{1}{1+\gamma}})$.

% For the second part of the remark, recall that $\vartheta(\bOne, \bfns (\bfth)) = \Nbar^{-1}\sum_{j\in\agents}\nstar_j (\bfth) \type_j$. By \eqref{eq:nstarpeculiar}, 
% \begin{align*}
%     \vartheta(\bOne, \bfns(\bfth)) &= \frac{1}{\Nbar}\sum_{j\leq \Lstar}\parenthese{\frac{\Nbar}{L}-\frac{2a}{c}(\empmean{L}-\type_j)}\type_j = \empmean{L} + \frac{2a}{c\Nbar}\sum_{j\leq \Lstar}\parenthese{\type_j ^2 - \empmean{L}\type_j} \\
%     &= \empmean{L} + \frac{2a}{c}\frac{L}{\Nbar}\empvar{L}\eqsp.
% \end{align*}
% By \Cref{lemma:excessrisk}, we obtain $\riskexcess(\bOne, \bfns(\bfth))=\alpha_{\delta}(1+\Nbar)^{-\gamma} + \beta + 2\empmean{L} + (4aL/c\Nbar)\empvar{L}$.
\end{proof}

\simplifiednegligiblewelfareloss*

\begin{proof}Recall that for any admissible $\bfn\in\Xi_{\collab}$, we can define $N=\bOne^{\transpose}\bfn$ and $\blamb=N^{-1}\bfn$ so that the social cost rewrites
$$
-W(\collab, \bfn) = -\widetilde{W}_{\collab}(\blamb, N)=f(N)+g(\blamb)\eqsp,
$$with $$f(N)= aJ(\Rstar + 2\alpha_{\delta}(1+N)^{-\gamma})+cN\quad\text{and}\quad g(\blamb)= 2aJ\blamb^{\transpose}\bfth\eqsp.$$By definition, $\sum_{j\in\agents}\nstar_j (\bfth)=\Nbar$ where $\Nbar = (\noutside + 1)J^{\frac{1}{1+\gamma}}-1=\argmin_{N\geq 0}f(N)$. Hence,

% Let $\Nbar=(\noutside+1)J^{\frac{1}{1+\gamma}}-1$,  which verifies $f(\Nbar)=\min_{N\geq 0}f(N)$, and  $\bar{\blamb}\in\Delta_J$ be such that $g(\Bar{\blamb})=\min_{\blamb\in\widetilde{\Xi}_{\collab}(\Nbar)}g(\blamb)$. $\bar{\blamb}$ exists because $g$ is continuous as a linear function, and $\widetilde{\Xi}_{\collab}(\Nbar)$ is compact as a closed subset of $\widetilde{\Xi}_\collab$, which is compact. Finally, denote by $N^\opt = \sum_{j\in\agents}\nopt (\bOne, \bfnopt (\bfth))$ and $\blamb^\opt = N^{\opt -1}\bfnopt (\bfth)$, where $\bfnopt (\bfth)$ is defined in \Cref{prop:optimalfullinfo}. 
\begin{align}
W(B^\opt , \bfn^\opt (\bfth))-W(\Bstar, \bfns(\bfth))&=\widetilde{W}_{\collab}(\blamb^\opt , N^\opt)-\widetilde{W}_{\collab} (\blamb^\star,\Nbar)\notag \\
&=f(\Nbar)+g(\blamb^\star)-f(N^\opt)-g(\blamb^\opt) \notag\\
&\leq g(\blamb^\star)-g(\blamb^\opt)=2aJ(\blamb^\star-\blamb^\opt)^{\transpose}\bfth \eqsp.\notag\\
\intertext{Since $\type_j \leq \type_{\Lstar}$ for any $j\in\agents$ such that $\lambda^\star _j > 0$, and $\type_j \geq \type_1$ for any $j\in\agents$ such that $\lambda^\opt _j > 0$:}
W(B^\opt , \bfn^\opt (\bfth))-W(\Bstar, \bfns(\bfth))&\leq 2aJ(\type_{\Lstar}-\type_1)\eqsp,\notag\\
\intertext{By assumption, for any $j\in\defEnsLigne{2,\ldots,J}$ $\type_j - \type_{j-1}=\bigO(1/J)$ so there exists $k_j > 0$ and $R_j \geq 0$ such that  $\type_j - \type_{j-1} \leq k_j J^{-1}$ for $J\geq R_j$. Denoting $R=\max_{j\in\agents}R_j$ and $k=\max_{j\in\agents}k_j$, for any $J\geq R$ we have $\type_{\Lstar}-\type_1 \leq k\Lstar$, so:}
W(B^\opt , \bfn^\opt (\bfth))-W(\Bstar, \bfns(\bfth))&\leq 2aJ\frac{k\Lstar}{J}=2ak\Lstar\eqsp,\notag
\end{align} so $W(B^\opt , \bfn^\opt (\bfth))-W(\Bstar, \bfns(\bfth))=\bigO(\Lstar)=\bigO(J^{\frac{1}{1+\gamma}})$ by   \Cref{cor:simplifyingfullinfo}.

    % Observe one the one hand that $u_j (\Bopt, \bfnopt(\bfth))=u_j (\Bstar, \bfns(\bfth))$ for any $j<\Lstar$, and on the other hand that  $\nstar_{\Lstar} (\bfth)=\maxcontrib^\star _{\Lstar} (\bfth)$ by definition. Thus, we have:
    % \begin{align*}
    %  W(\Bopt, \bfnopt(\bfth))-W(\collab^\star, \bfns(\bfth))&=u_{\Lstar} (\Bopt, \bfnopt(\bfth))-u_{\Lstar} (\Bstar, \bfns(\bfth))   \\
    %  &= a(\eps(\Bstar,\bfns(\bfth))-\eps(\Bopt, \bfnopt(\bfth)))+c(\maxcontrib^{\star}_{\Lstar} (\bfth) - \nopt_{\Lstar} (\bfth)) \\
    %  &\leq  a(\eps(\Bstar,\bfns(\bfth))-\eps(\Bopt, \bfnopt(\bfth)))+c\maxcontrib^\star _{\Lstar} (\bfth) \\
    %  &=-\outside_{\Lstar} -a\eps(\Bopt, \bfnopt(\bfth)) \\
    %  &=a\parentheseDeux{\Rstar + \eps(\type_{\Lstar},\noutside)-\eps(\Bopt, \bfnopt(\bfth))}+c\noutside \\
    %  &< a(\Rstar +\typemax-\typemin) + c\noutside =o(J)\eqsp.
    % \end{align*}
\end{proof}
% \simplifiednegligiblewelfareloss*

% \begin{proof}
%     Observe that by definition of $(\collab^\star, \bfns)$ and $\nstar_L (\bfth)=\maxcontrib^\star _L (\bfth)$,
%     \begin{align*}
%      W(\Bopt, \bfnopt(\bfth))-W(\collab^\star, \bfns(\bfth))&=u_L (\Bopt, \bfnopt(\bfth))-u_L (\Bstar, \bfns(\bfth))   \\
%      &= 2a(\eps(\Bstar,\bfns(\bfth))-\eps(\Bopt, \bfnopt(\bfth)))+c(\nstar_L (\bfth) - \nopt_L (\bfth)) \\
%      &\leq 2a(\eps(\Bstar,\bfns(\bfth))-\eps(\Bopt, \bfnopt(\bfth)))+c\maxcontrib^\star _L (\bfth) \\
%      &= a(\eps(\type_L, \noutside)-\eps(\Bopt, \bfnopt(\tbfth)))+c\noutside \\
%      &< a(\typemax-\typemin) + c\noutside =o(J)\eqsp.
%     \end{align*}
% \end{proof}
\unravelling*
\begin{proof} First, we show that the situation where the coalition is empty is a Nash equilibrium. Consider $\bfs=\parentheseLigne{(0,\dagger),\ldots,(0,\dagger)}\in\strategyspace^N$. For any $j\in\agents$ and deviation $(1,\ttype_j)$, $\cB=\defEns{j}$ so $\collabmodel=\agentmodel{j}$ and $v_j((0,\dagger), \bsstar_{-j})=\outside _j \geq v_j((1,\ttype_j), \bfs_{-j})$ by \propref{prop:outsideoption}. Thus, $\sstar\in\nasheqspace$.

Now, consider a pure-strategy Nash equilibrium $\sstar\in\nasheqspace$ such that $\cB=\set{j\in\agents}{B_j = 1}\ne \emptyset$. Denote by $\cC=\defEnsLigne{j\in\cB:\: \nstar_j (\tbfth)>0}$ the set of contributors. We start with two technical lemmas:

\begin{lemma}
    \label{lemma:deltatypebid}There exists $\Delta\in\R$ such that for any $(j,k)\in\cC^2$
    $$
\type_{j}-\ttype_{j}=\type_{k}-\ttype_{k}=\Delta\eqsp.
    $$
\end{lemma}

\begin{proof}
      To lighten notation, we write in this proof $$\vartheta(\tbfth)=N^{-1}\sum_{j\in\cC}\nstar_j (\tbfth)\,\type_j\quad\text{and}\quad\tvt(\tbfth)=N^{-1}\sum_{j\in\cC}\nstar_j (\tbfth)\,\ttype_j,\quad\text{where}\quad N=\sum_{k\in\cC}\nstar_k (\tbfth)\eqsp.$$By \Cref{def:simplifiedscheme} and \Cref{assumption:simplifiedscheme}, for any $j\in\cC$, $\nstar_{j}(\tbfth) = \maxcontrib_{j}(\collab, \bfns(\tbfth))=\noutside - \frac{a}{c}\parentheseLigne{\eps(\tvt(\tbfth), N) - \eps(\ttype_{j},\noutside)}$, so their payoff reads
\begin{align}  
    v_{j}((1,\ttype_{j}), \bfs_{-{j}})
    &= -a(\Rstar + \eps(\vartheta(\tbfth), N) - c\parentheseDeux{\noutside - \frac{a}{c}\parenthese{\eps(\tvt(\tbfth), N) - \eps(\ttype_{j},\noutside)}} \notag\\
    &= -a(\Rstar + \eps(\type_{j}, \noutside)) -c\noutside \notag\\
    &+ a\parentheseDeux{\parenthese{\eps(\type_j , \noutside)-\eps(\ttype_j , \noutside)} - \parenthese{\eps(\vartheta(\tbfth), N)-\eps(\tvt(\tbfth), N)}}\notag\\
    &= \outside _{j} + 2a\parentheseDeux{\parenthese{\type_{j} - \ttype_{j}}-\parenthese{\vartheta(\tbfth) - \tvt(\tbfth)}}\label{eq:typesminusvartheta}\eqsp.
\end{align}Since $\bfs$ is a Nash equilibrium and $\outside_{j}=v_{j}((0,\dagger),\bfs_{-j})$, we have in particular that
$$
2a\parentheseDeux{\parenthese{\type_{j} - \ttype_{j}}-\parenthese{\vartheta(\tbfth) - \tvt(\tbfth)}}\geq 0\eqsp,$$that is $ \type_{j}-\ttype_{j}=\Delta_{j}\geq \Delta=\vartheta(\tbfth) - \tvt(\tbfth)$ for any $j\in\cC$. We now show that this holds with strict equality for any $j\in\cC$. By contradiction, assume there exists $r\in\cC$ such that $\Delta_{r}=\Delta + \chi$ with $\chi >0$. Then\begin{align*}
    \Delta = \vartheta(\tbfth) - \tvt(\tbfth) = \sum_{j\in\cC}\lambda^\star _{j} (\tbfth) \Delta_{j}&=\sum_{j\in\cC\setminus\defEnsLigne{r}}\lambda^\star _{j} (\tbfth)\Delta_{j} + \lambda^\star _{r} (\tbfth)(\Delta + \chi) \\
    &\geq \Delta + \lambda^\star _{r} (\tbfth)\chi > \Delta\eqsp,
\end{align*}because $\lambda^\star _{r} (\tbfth)>0$ as $r\in\cC$. This is a contradiction, and establishes the result. \end{proof}

\begin{lemma}
    \label{lemma:typeminstrictlydominant}
    For any $j\in \cC$, $\ttype_{j}=\typemin\eqsp.$
\end{lemma}
\begin{proof}

Let $j\in\cC$, by \eqref{eq:typesminusvartheta}:
 \begin{align*}
     v_{j}((1,\ttype_{j}), \bfs_{-j})  &=\outside _{j} + 2a\parentheseDeux{\parenthese{\type_{j} - \ttype_{j}}-\parenthese{\vartheta(\tbfth) - \tvt(\tbfth)}} \\
     &=\outside_{j} + 2a\parentheseDeux{\parenthese{1-\lambda_{j} ^{\star}(\tbfth)}\parenthese{\type_{j} - \ttype_{j}} - \sum_{k\in\cC\setminus\defEnsLigne{j} }\lambda_{k} ^{\star}(\tbfth)\parenthese{\type_{k} - \ttype_{k}}}\notag \\
    &= \outside_{j} + 2a\parentheseDeux{\parenthese{\sum_{k\in\cC\setminus\defEnsLigne{j}}\lambda^\star _{k}(\tbfth)} \parenthese{\type_{j} - \ttype_{j}}-\sum_{k\in\cC\setminus\defEnsLigne{j}}\lambda^\star _{k}(\tbfth) \parenthese{\type_{k} - \ttype_{k}}}\notag\\
&=\outside_{j} + 2a q_j (\ttype_j , \tbfth_{-j})\eqsp,
 \end{align*}

where\begin{equation}\label{eq:utilitycontributor}
    q_j(\ttype_j , \tbfth_{-j})= \sum_{k\in\cC\setminus\defEnsLigne{j}}\lambda^\star_{k} (\tbfth)\parentheseDeuxLigne{(\type_{j}-\ttype_{j})-(\type_{k}-\ttype_{k})}\eqsp.
\end{equation}

We prove that $v_{j}((1,\centraldot), \bfs_{-j})$ is strictly decreasing in $\ttype_{j}$ for $j\in\cC$, by showing that $\partial \hv_{j}((1,\ttype_{j}),\bfs_{-{i_j}})/\partial \ttype_{j}< 0$ For any $j\in\cC$. For  $v_{j}$ to be differentiable, we need $q_j$ to be differentiable, that is $\lambda^\star _{j} (\tbfth)=(\sum_{k\in\cC}\nstar_{k}(\tbfth))^{-1}\nstar_{j}(\tbfth)$ to be differentiable for any $j\in\cC$. We re-index $\cC$ as $\defEnsLigne{i_1 , \ldots, i_{\abs{\cC}}}$ so that $\ttype_{i_1}<\ldots < \ttype_{i_{\abs{\cC}}}$. Observe that for $0<h<\delta$ for $\delta >0$ sufficiently small, $\ttype_{i_{j-1}}<\ttype_{j}+h<\ttype_{i_{j+1}}$ because $\ttype_{i_{j-1}}<\ttype_{j}<\ttype_{i_{j+1}}$ by \Cref{lemma:deltatypebid} and \Cref{assumption_types_sorted}. Hence, the bid ordering does not change for any infinitesimal variation $\rmd\ttype_{j}>0$, nor does the indicator $\1\defEnsLigne{i_j\leq i_{\abs{\cC}}}$. Consequently by     \Cref{cor:simplifyingfullinfo}, $\nstar_{j}(\tbfth)$ is differentiable in $\ttype_{j}$ and so is $\lambda^\star (\tbfth)$. We have for any $j\in\cC$:
 \begin{align}
     \frac{\partial v_{j}((1,\ttype_{j}), \bfs_{k})}{\partial \ttype_{j}} &=2a\frac{\partial q_j (\ttype_j , \tbfth_{-j})}{\partial \ttype_{j}}\notag \\
     &=2a\sum_{k\in\cC\setminus\defEnsLigne{j}}\parentheseDeux{\frac{\partial \lambda^\star _{k}(\tbfth)}{\partial \ttype_{j}}\parenthese{\parenthese{\type_{j}-\ttype_{j}}-\parenthese{\type_{k}-\ttype_{k}}} -\lambda^\star _{k}(\tbfth)}\notag \\
    &= 2a\sum_{k\in\cC\setminus\defEnsLigne{j}}\frac{\partial \lambda^\star _{k}(\tbfth)}{\partial \ttype_{j}}\parenthese{\Delta - \Delta}-2a\sum_{k\in\cC\setminus\defEnsLigne{j}}\lambda^\star _{k}(\tbfth) \\
    &<0\eqsp.\label{eq:derivativenegative}
\end{align}
where we have used \Cref{lemma:deltatypebid} at the third line. This implies
\begin{equation}
    \label{eq:typemindominantforlminusone}
    \ttype_{j}=\typemin\qquad\text{for any } j\in\cC\eqsp.
\end{equation}
To see why, assume by contradiction that there exists $j\in\cC$ such that $s_{j} = (1,\ttype_{j})$ with $\ttype_{j} > \typemin$. Then for any $h\in(0,\ttype_{j}-\typemin]$, by  \eqref{eq:derivativenegative}: 
\begin{align*}
    v_{j}((1,\ttype_{j}-h), \bfs_{-j}) = v_{j}((1,\ttype_{j}), \bfs_{-j}) - \int_{\ttype_{j}-h}^{\ttype_{j}} \frac{\partial v_{j}((1,t), \bfs_{k})}{\partial t}dt > v_{j}((1,\ttype_{j}), \bfs_{k})\eqsp,
\end{align*}
which contradicts $\bfs$ being a Nash equilibrium.\end{proof}
Combining \Cref{lemma:deltatypebid} and  \Cref{lemma:typeminstrictlydominant} gives for any $(j,k)\in\cC^2$:
\begin{equation}
    \label{eq:contributorshavesametype}
    \type_{j}=\type_{k}\eqsp.
\end{equation}
Recall that by \Cref{assumption_types_sorted}, $\type_{m}=\type_{n}$ if and only if $m=n$, so \eqref{eq:contributorshavesametype} implies $\abs{\cC}= 1$. This in turn implies $\abs{\cB}= 1$. 
Indeed, assume by contradiction $\abs{\cC} = 1$ and  $\abs{\cB}\geq 2$. Denote $r\in\cB$ the only contributing agent. They are asked $\nstar_r (\tbfth) = \maxcontrib_r (\tbfth)=\noutside$. Moreover by \eqref{eq:defnbar} $\Nbar=\abs{\cB}^{1/1+\gamma}(\noutside + 1) - 1  > \noutside$, so by definition of the contribution scheme \eqref{def:simplifiedscheme}, there exists $k\in\cB\setminus\defEns{r}$ such that $\nstar_k (\tbfth)>0$. This contradicts $\abs{\cC}=1$.

We now show that $\cB = \defEns{J}$. By contradiction, assume $\cB=\defEns{j}$ with $j <J$. In particular, $s_J = (0,\dagger)$. Consider a deviation $s'_J = (1,\ttype_{J})$ with $\ttype_J \in\types$, so $\cB = \defEns{j,J}$ under $(s'_{J},\bfs_{-J})$. By \Cref{lemma:typeminstrictlydominant}, $\ttype_j = \ttype_J = \typemin$, so \eqref{eq:typesminusvartheta} rewrites:
\begin{equation}
    \label{eq:utilityweightedaverage}
    v_{k}((1,\ttype_k),\bfs_{-k})=\outside_k + 2a (\type_k - \vartheta(\collab, \bfns(\tbfth))\eqsp,
\end{equation}for any $k\in\defEns{j,J}$. Since $\type_j < \vartheta(\collab, \bfns(\tbfth))<\type_J$ by \Cref{assumption_types_sorted}, we have
$$
v_J((1,\ttype_J),\bfs_{-J})>\outside_J = v_J ((0,\dagger),\bfs_{-J})\eqsp,
$$
which contradicts $\bfs$ begin a Nash equilibrium. Hence, $\cB=\defEns{J}$. This concludes the proof. 
\end{proof}

\positivepaymentvcg*

\begin{proof}
    Let $(\collab,\tbfth)\in\defEns{0,1}^{J}\times\types^N$ be an equilibrium of the game induced by $\Gamma^\textsc{VCG}$. Since the VCG mechanism is strategyproof in dominant strategy, $\tbfth=\bfth$.  For any $j\in\agents$, the VCG payment is
    \begin{align}
        t_j(\tbfth) &= \sum_{k\ne j}u_k((0,\bOne_{-j}), \bfns(\bfth))-\sum_{k\ne j}u_k (\bOne, \bfns(\bfth))\notag\\
        &=\underbracket[0.140ex]{W((0,\bOne_{-j}),\bfns(\bfth))-W(\bOne,\bfns(\bfth))}_{\mathbf{(I)}}-\parentheseDeuxLigne{\underbracket[0.140ex]{u_j ((0,\bOne_{-j}), \bfns (\bfth))-u_j (\bOne, \bfns(\bfth))}_{\mathbf{(II)}}}\eqsp.\label{eq:vcgpayment}
    \end{align}
    Consider $j\in\iint{1}{\Lstar}$. By the contribution scheme \Cref{cor:simplifyingfullinfo}, $$u_j (\bOne, \bfns(\bfth))=-a(\Rstar + \eps(\bOne, \bfns(\bfth)))-c\nstar_j (\bfth)=\outside_j = u_j((0,\bOne_{-j}),\bfns(\bfth))\eqsp,$$where we used in the third inequality that $\nstar_j (\bfth)=\maxcontrib_j (\collab, \bfns(\bfth))$. Thus,  $\mathbf{(II)}=0$. Moreover by \Cref{prop:optimalfullinfo}, $W((0,\bOne_{-j}),\bfns(\bfth))<W(\bOne,\bfns(\bfth))$ so $\mathbf{(I)}<0$. Thus, by \eqref{eq:vcgpayment} we have $t_j (\tbfth)<0$.
\end{proof}

\truthfulnessisnash*

\begin{proof}Let $\collab^\star = (1,\ldots,1)^{\transpose}$, so the  aggregator has at her disposal $\htype_1, \ldots, \htype_J$.
% We sort $\agents$ as $\defEnsLigne{i_1, \ldots, i_J}$ so that $\htype_{i_1} \leq\ldots\leq\htype_{i_J}$. We define $C=\defEnsLigne{j\in\agents:\nstar_j (\bfth)>0}$ as well as $\hC=\defEnsLigne{j\in\agents:\: \nstar_j (\htype_j + \etadl , \hbfth_{-j})>0}$, which are respectively the set of contributors with the actual types and the estimated types. We also write $\Lstar=\sum_{j\in\agents}\1\defEns{j\in C}$ and $\Lhat=\sum_{j\in\agents}\1\defEns{j\in\widehat{C}}$.

We first focus on the event under which $\htype_j$ correctly approximates $\type_j$ for any $j\leq J$. Define $\sfE_j (\delta)=\defEnsLigne{\omega\in\Omega:\: \absLigne{\htype_j(\omega) - \type_j}\leq \eta_{\delta}}$ and  $\sfE(\delta)=\bigcap_{j\leq J}\sfE_j (\delta)$. We have 
\begin{equation}
    \label{lemma:probaEdeltaL}
     \P(\sfE(\delta / J))\geq 1-\delta\eqsp,
\end{equation}
because $\P(\sfE(\delta))=\P(\cap_{j\leq J}\sfE_j (\delta))=1-\P(\cup_{j\leq J}\sfE_j (\delta)) \geq 1-\sum_{j\leq  J}\P(\sfE_j (\delta))=1-J\delta$, and setting $\delta' = J\delta$ leads to \Cref{lemma:probaEdeltaL}. In what follows, we assume that $E(\delta/J)$ is true so $\absLigne{\htype_j - \type_j}\leq \etadl$ for any $j\in\agents$. We denote by $\cC=\defEnsLigne{j\in\agents:\: m_j (\hbfth)=\nstar_j (\hbfth+\bfeta_j )>0}$ the set of contributors, and to lighten notations we write
    \begin{align*}
        \vartheta(\hbfth)&=\Nbar^{-1}\bfns(\hbfth+\bfeta_j)^{\transpose}\bfth = \blamb^\star (\hbfth + \bfeta_j)^{\transpose}\bfth\\
\text{and}\quad\hvt_j(\hbfth)&=\Nbar^{-1}\bfns(\hbfth + \bfeta_j)^{\transpose}(\hbfth + \bfeta_j)=\blamb^\star (\hbfth + \bfeta_j)^{\transpose}(\hbfth + \bfeta_j)\eqsp,\\
\text{where}\quad\bfeta_j &= \etadl \bOne - 2\boldsymbol{\delta}_j \etadl,\quad\text{with}\quad\boldsymbol{\delta}_j = (0,\ldots,0,1,0,\ldots,0)\eqsp,
    \end{align*}
for any $j\in\agents$. Note that $\vartheta (\hbfth)$ is the actual weighted type within the coalition, whereas $\hat{\vartheta}(\hbfth)$ is the weighted type estimated by the aggregator. Consider first a contributor $j\in\cC$ who is asked $\nstar_j (\hbfth +\bfeta_j)>0$ samples. Their payoff under $\Bstar$ is
    \begin{align}
        \hv_j (1,\bOne_{-j})&=-a(\Rstar + \eps(\vartheta(\hbfth), N))-c\nstar_j (\hbfth + \bfeta_j) \notag\\
        &= -a(\Rstar + \eps(\vartheta(\hbfth), N)) -c\parentheseDeux{\noutside - \frac{a}{c}\parenthese{\eps(\hvt_j(\hbfth),N)-\eps(\htype_j -\etadl, \noutside)}}\notag\\
        &= -a\parenthese{\Rstar + \eps(\type_j , \noutside)}-c\noutside \notag\\
        &+a\parentheseDeux{\parenthese{\eps(\type_j, \noutside)-\eps(\htype_j - \etadl)} - \parenthese{\eps(\vartheta(\hbfth),N)-\eps(\hvt_j(\hbfth),N)}} \notag\\
        &=\outside_j + 2a\parentheseDeux{\parenthese{\type_j - (\htype_j - \etadl)}-\parenthese{\vartheta(\hbfth)-\hvt_j (\bfth)}}\notag\\
        &=\hv_j (0,\bOne_{-j})+2a(1-\lambda^\star _j (\hbfth + \bfeta_j))(\type_j -(\htype_j - \etadl))\notag\\
        &-2a\sum_{\ell \in\cC\setminus\defEnsLigne{j}}\lambda^\star _{\ell}(\hbfth + \bfeta_j)(\type_{\ell}-(\htype_{\ell}+\etadl))\eqsp,\notag\\
        \intertext{and since $E(\delta/J)$ is true, $\type_j \geq \htype_j - \etadl$ and $\type_{\ell} \leq \htype_{\ell} + \etadl$ for any $\ell \in\cC\setminus\defEnsLigne{j}$, so}
        \hv_j (1,\bOne_{-j})&\geq \hv_j (0,\bOne_{-j})\eqsp.\label{eq:casejinc}
    \end{align}Second, consider $j\notin \cC$ who is asked $ \qubar \leq \noutside - 2(a/c)(\typemax-\typemin) $ samples, and denote by $r\in\agents$ the agent such that $\type_r = \max_{k\in\cC}\type_k$. Observe that $\qubar \leq \nstar_r (\hbfth+\bfeta_r)$, so
    \begin{align*}
        \hv_j (1,\bOne_{-j})&=-a(\Rstar + \eps(\vartheta(\hbfth),N))-c\qubar \\
        &\geq -a(\Rstar + \eps(\vartheta(\hbfth),N))-c\nstar_r (\hbfth + \bfeta_r)\\
        &\geq \outside_r\eqsp,
        \intertext{where the last inequality comes from $r\in\cC$ and \eqref{eq:casejinc}. Since $j\notin\cC$, we have $\type_j \geq\htype_j - \etadl \geq \htype_r + \etadl \geq \type_r $ so $\outside_r \geq \outside_j$, and it follows that}
        \hv_j (1,\bOne_{-j}) &\geq \outside_j = \hv_j (0, \bOne_{-j})\eqsp.
    \end{align*}\end{proof}

\estimatorissuitable*

\begin{proof}
    Let $j\in\agents$ and $g\in\hypoclass$, we have
    \begin{align}
        \htype^{\textsc{erm}}_{0,j}&=\sup_{g\in\hypoclass}\abs{\emprisk{g}{j} - \emprisk{g}{0}} \notag \\
        &\leq \sup_{g\in\hypoclass}\abs{\emprisk{g}{j} - \risk{g}{j}} + \sup_{g\in\hypoclass}\abs{\risk{g}{j}-\risk{g}{0}} + \sup_{g\in\hypoclass}\abs{\emprisk{g}{0} - \risk{g}{0}} \notag \\
        & \leq \alpha_{\delta/4}\parentheseDeux{(q+1)^{-\gamma}+(\qprime+1)^{-\gamma}}
        + 2\beta  + \type_j\eqsp,\label{eq:boundonesense} \\
    \end{align}with probability $1-\delta/2$ by \Cref{definition_type}, \Cref{assumption_upper_bound_risk},  \Cref{assumption_sample_from_Pzero}, and a union bound. Similarly, 
        \begin{align}
        \type_j&=\sup_{g\in\hypoclass}\abs{\risk{g}{j} - \risk{g}{0}} \notag\\
        &\leq \sup_{g\in\hypoclass}\abs{\risk{g}{j} - \emprisk{g}{j}} + \sup_{g\in\hypoclass}\abs{\emprisk{g}{j}-\emprisk{g}{0}} + \sup_{g\in\hypoclass}\abs{\emprisk{g}{0} - \risk{g}{0}}\notag \\
        & \leq \alpha_{\delta/4}\parentheseDeux{(q+1)^{-\gamma}+(\qprime+1)^{-\gamma}}
        + 2\beta  + \htype^{\textsc{erm}}_{0,j}\eqsp,\label{eq:boundothersense}\eqsp
    \end{align}
    with probability $1-\delta/2$. Combining \eqref{eq:boundonesense} and \eqref{eq:boundothersense} along with an union bound yields
    $$
\absLigne{\htype_{0,j} ^{\textsc{erm}}-\type_{j}} \leq\alpha_{\delta/4}\parentheseDeux{(q+1)^{-\gamma}+(\qprime+1)^{-\gamma}}
        + 2\beta\eqsp,
    $$with probability $1-\delta$.
\end{proof}

\exampleclassif*

\begin{proof}\begin{enumerate}[wide, labelwidth=!, labelindent=0pt, label=(\roman*)]\item Let $j\in\cB$. Observe that under \Cref{assumption:classification},we have $$\emprisk{-g}{j}=n_j ^{-1}\sum_{i=1}^{n_j}\1\defEnsLigne{-g(X^j _i)Y^j _i <0}=n_j^{-1}\sum_{i=1}^{n_j} (1-\1\defEnsLigne{g(X^j _i)Y^j _i<0})=1-\emprisk{g}{j}\eqsp.$$Since the hypothesis class $\hypoclass$ is symmetric, we have
\begin{align*}
    \htype^{\textsc{erm}}_j &= \sup_{g\in\hypoclass}\abs{\emprisk{g}{j}-\emprisk{g}{0}} = \sup_{g\in\hypoclass}\parenthese{\emprisk{g}{j}-\emprisk{g}{0}}\\
    &=\sup_{g\in\hypoclass}\parenthese{1-\parenthese{\emprisk{-g}{j}+\emprisk{g}{0}}}
    = 1-\inf_{g\in\hypoclass}\parenthese{\emprisk{g}{j^{-}}+\emprisk{g}{0}} \eqsp.
\end{align*}
\item  Let $j\in\agents$, we have   \begin{align*}
\eta_{\delta/J}(q)&=\alpha_{\delta/4J}[(q+1)^{-\gamma}+(1+\qprime)^{-\gamma}]+2\beta\\
&=\ln(4J / \delta)^{1/2}[(1+q)^{-\gamma}+(1+\qprime)^{-\gamma}] + 2\mathrm{Rad}(\hypoclass)\eqsp.
    \end{align*}
\end{enumerate}
\end{proof}

\newpage
\section*{NeurIPS Paper Checklist}

%%% END INSTRUCTIONS %%%

\begin{enumerate}

\item {\bf Claims}
    \item[] Question: Do the main claims made in the abstract and introduction accurately reflect the paper's contributions and scope?
    \item[] Answer: \answerYes{} % Replace by \answerYes{}, \answerNo{}, or \answerNA{}.
    \item[] Justification: the claimed contributions are supported by proven theorems.

\item {\bf Limitations}
    \item[] Question: Does the paper discuss the limitations of the work performed by the authors?
    \item[] Answer: \answerYes{} % Replace by \answerYes{}, \answerNo{}, or \answerNA{}.
    \item[] Justification: all the made assumptions are clearly highlighted and discussed.

\item {\bf Theory Assumptions and Proofs}
    \item[] Question: For each theoretical result, does the paper provide the full set of assumptions and a complete (and correct) proof?
    \item[] Answer: \answerYes{} % Replace by \answerYes{}, \answerNo{}, or \answerNA{}.
    \item[] Justification: Assumptions are made clear and given for each theorem individually. Full proofs are given in the Appendix.

    \item {\bf Experimental Result Reproducibility}
    \item[] Question: Does the paper fully disclose all the information needed to reproduce the main experimental results of the paper to the extent that it affects the main claims and/or conclusions of the paper (regardless of whether the code and data are provided or not)?
    \item[] Answer: \answerNA{} % Replace by \answerYes{}, \answerNo{}, or \answerNA{}.
    \item[] Justification: no experimental result in the paper

\item {\bf Open access to data and code}
    \item[] Question: Does the paper provide open access to the data and code, with sufficient instructions to faithfully reproduce the main experimental results, as described in supplemental material?
    \item[] Answer: \answerNA{} % Replace by \answerYes{}, \answerNo{}, or \answerNA{}.
    \item[] Justification: no experimental result in the paper

\item {\bf Experimental Setting/Details}
    \item[] Question: Does the paper specify all the training and test details (e.g., data splits, hyperparameters, how they were chosen, type of optimizer, etc.) necessary to understand the results?
    \item[] Answer: \answerNA{} % Replace by \answerYes{}, \answerNo{}, or \answerNA{}.
    \item[] Justification: no experimental result in the paper

\item {\bf Experiment Statistical Significance}
    \item[] Question: Does the paper report error bars suitably and correctly defined or other appropriate information about the statistical significance of the experiments?
    \item[] Answer: \answerNA{} % Replace by \answerYes{}, \answerNo{}, or \answerNA{}.
    \item[] Justification: no experimental result in the paper

\item {\bf Experiments Compute Resources}
    \item[] Question: For each experiment, does the paper provide sufficient information on the computer resources (type of compute workers, memory, time of execution) needed to reproduce the experiments?
    \item[] Answer: \answerNA{} % Replace by \answerYes{}, \answerNo{}, or \answerNA{}.
    \item[] Justification: no experimental result in the paper
    
\item {\bf Code Of Ethics}
    \item[] Question: Does the research conducted in the paper conform, in every respect, with the NeurIPS Code of Ethics \url{https://neurips.cc/public/EthicsGuidelines}?
    \item[] Answer: \answerYes{} % Replace by \answerYes{}, \answerNo{}, or \answerNA{}.
    \item[] Justification: no data used for this work

\item {\bf Broader Impacts}
    \item[] Question: Does the paper discuss both potential positive societal impacts and negative societal impacts of the work performed?
    \item[] Answer: \answerNA{} % Replace by \answerYes{}, \answerNo{}, or \answerNA{}.
    \item[] Justification: no societal impact of the work performed
    
\item {\bf Safeguards}
    \item[] Question: Does the paper describe safeguards that have been put in place for responsible release of data or models that have a high risk for misuse (e.g., pretrained language models, image generators, or scraped datasets)?
    \item[] Answer: \answerNA{} % Replace by \answerYes{}, \answerNo{}, or \answerNA{}.
    \item[] Justification: no data

\item {\bf Licenses for existing assets}
    \item[] Question: Are the creators or original owners of assets (e.g., code, data, models), used in the paper, properly credited and are the license and terms of use explicitly mentioned and properly respected?
    \item[] Answer: \answerNA{} % Replace by \answerYes{}, \answerNo{}, or \answerNA{}.
    \item[] Justification: no existing assets

\item {\bf New Assets}
    \item[] Question: Are new assets introduced in the paper well documented and is the documentation provided alongside the assets?
    \item[] Answer: \answerNA{} % Replace by \answerYes{}, \answerNo{}, or \answerNA{}.
    \item[] Justification: no new asset

\item {\bf Crowdsourcing and Research with Human Subjects}
    \item[] Question: For crowdsourcing experiments and research with human subjects, does the paper include the full text of instructions given to participants and screenshots, if applicable, as well as details about compensation (if any)? 
    \item[] Answer: \answerNA{} % Replace by \answerYes{}, \answerNo{}, or \answerNA{}.
    \item[] Justification: no crowdsourcing nor research with human subjects

\item {\bf Institutional Review Board (IRB) Approvals or Equivalent for Research with Human Subjects}
    \item[] Question: Does the paper describe potential risks incurred by study participants, whether such risks were disclosed to the subjects, and whether Institutional Review Board (IRB) approvals (or an equivalent approval/review based on the requirements of your country or institution) were obtained?
    \item[] Answer: \answerNA{} % Replace by \answerYes{}, \answerNo{}, or \answerNA{}.
    \item[] Justification: no crowdsourcing nor research with human subjects

\end{enumerate}
\end{document}